\documentclass[onecolumn]{IEEEtran}

\usepackage{amsmath}   
\usepackage{amssymb}

\RequirePackage{amsthm,amsmath,amssymb,amsfonts,graphicx, subfigure}
\RequirePackage[colorlinks]{hyperref}

\theoremstyle{plain} 
\newtheorem{thm}{Theorem}[section]
\newtheorem{prop}[thm]{Proposition}
\newtheorem{cor}[thm]{Corollary}
\newtheorem{lem}[thm]{Lemma}
\newtheorem{defn}[thm]{Definition}
\newtheorem{conj}[thm]{Conjecture}

\newtheorem*{informal}{Main Theorem}

\theoremstyle{remark} 
\newtheorem{rmk}{Remark}

\newcommand{\bc}{\begin{center}}
\newcommand{\ec}{\end{center}}
\newcommand{\bt}{\begin{tabular}}
\newcommand{\et}{\end{tabular}}
\newcommand{\bea}{\begin{eqnarray}}
\newcommand{\eea}{\end{eqnarray}}

\newcommand{\ba}{\begin{array}}
\newcommand{\ea}{\end{array}}

\def\be{\begin{eqnarray}}
\def\ee{\end{eqnarray}}
\def\ben{\begin{eqnarray*}}
\def\een{\end{eqnarray*}}

\newcommand{\levy}{\mbox{L\'evy }}

\newcommand{\ra} {\rightarrow}

\newcommand{\nth}{\frac{1}{n}}

\newcommand{\RL}{{\mathbb R}}

\def\sq{$\Box$}

\def\qed{\ifmmode\sq\else{\unskip\nobreak\hfil
\penalty50\hskip1em\null\nobreak\hfil\sq
\parfillskip=0pt\finalhyphendemerits=0\endgraf}\fi\par\medbreak}

\def\tr{{\rm tr\, }}

\newsavebox{\junk}
\savebox{\junk}[1.6mm]{\hbox{$|\!|\!|$}}

\def\limsup{\mathop{\rm lim\ sup}}
\def\liminf{\mathop{\rm lim\ inf}}

\def\half{{\mathchoice{\textstyle \frac{1}{2}}%
{\frac{1}{2}}%
{\hbox{\tiny $\frac{1}{2}$}}%
{\hbox{\tiny $\frac{1}{2}$}} }}

 \def\eq#1/{(\ref{#1})}

\def\eq#1/{(\ref{e:#1})}
\def\E{{\bf E}}

\def\phi{\varphi}

\def\bee{\begin{eqnarray*}}
\def\ene{\end{eqnarray*}}

\begin{document}
\title{Beyond the entropy power inequality, via rearrangements}
\author{Liyao Wang
	\thanks{L. Wang is with the Department of Physics, Yale University,
	217 Prospect Street, New Haven, CT 06511, USA.
	Email: {\tt liyao.wang@yale.edu}
	}
 and 
Mokshay~Madiman,~\IEEEmembership{Member,~IEEE}
	\thanks{M. Madiman is with the Department of Mathematical Sciences, University of Delaware,
	517B Ewing Hall, Newark, DE 19716, USA.
	Email: {\tt madiman@udel.edu}
	}
}
\footnotetext{This research was supported by the U.S. National Science Foundation through the grants 
DMS-1409504 and CCF-1346564.}
\footnotetext{Portions of this paper were presented at the 2013 Information Theory and Applications Workshop
in San Diego, and at the 2013 IEEE International Symposium on Information Theory in Istanbul.}%

\maketitle

\begin{abstract}
A lower bound on the R\'enyi differential entropy of a sum of independent random vectors is
demonstrated in terms of rearrangements. For the special case of Boltzmann-Shannon entropy,
this lower bound is better than that given by the entropy power inequality. Several
applications are discussed, including a new proof of the classical entropy power inequality
and an entropy inequality involving symmetrization of L\'evy processes.
\end{abstract}


%


\section{Introduction}
\label{sec:intro}

Rearrangement is a natural and powerful notion in functional analysis,
and finds use in proving many interesting inequalities.
For instance, the original proofs of Young's inequality with
sharp constant (which, as is well known from \cite{DCT91},
is a common generalization of the Brunn-Minkowski and
entropy power inequalities) rely on rearrangements \cite{Bec75, BL76b}.
A basic property of rearrangements is that they preserve $L^p$ norms;
thus, in particular, the rearrangement of a probability density function
is also a probability density function.

Our main contribution in this note is a new lower bound on the R\'enyi (differential) entropy
of a sum of independent random vectors taking values in $\RL^n$, for some fixed natural number $n$.
Recall that for $p\in(0,1)\cup(1,+\infty)$, the R\'enyi entropy of a probability density $f$ is defined as:
\ben
h_p(f)=\frac{1}{1-p}\log\bigg(\int_{\RL^n} f^p(x)dx\bigg).
\een
For $p=1$, $h_1(f)$ is defined as the Shannon differential entropy
\ben
h(f)=-\int f(x)\log f(x) dx ,
\een
and for $p=0,\infty$, it is defined in a limiting sense (see Lemma~\ref{lem:key} for details).

This new bound is expressed
in terms of rearrangements, which we define
and recall basic properties of in Section~\ref{sec:rearr-basic}.

\begin{informal}
If $f_i$ are densities on $\RL^n$,
and $f_i^*$ are their spherically symmetric rearrangements,
\be\label{eq:rear1}
h_p(f_1\star f_2\star\cdot\cdot\star f_k) \geq h_p(f_1^* \star f_2^*\star\cdot\cdot f_k^*) ,
\ee
for any $p\in[0,1)\cup(1,\infty]$. For $p=1$,
\be\label{eq:rear2}
h(f_1\star f_2\star\cdot\cdot\star f_k)\geq h(f_1^* \star f_2^*\star\cdot\cdot f_k^*),
\ee
provided both sides are well defined.
\end{informal}

If we write $X_i^*$ for a random vector drawn from the density $f_i^*$,
and assume that all random vectors are drawn independently of each other,
the Main Theorem says in more customary information-theoretic notation that
\ben
h_p(X_1 +\ldots+ X_k)\geq h_p(X_1^* +\ldots+ X_k^*)
\een
for each $p\geq 1$.

Let us note that the special case of the Main Theorem corresponding to $p\in (0,1)$ and $k=2$
is implicit in \cite[Proposition 9]{BL76b}.
However, our extension includes the three most interesting values of $p$ (namely,
0, 1, and $\infty$), and also covers arbitrary positive integers $k$.
Indeed, as we will discuss, the $p=0$ case yields the Brunn-Minkowski
inequality, the $p=\infty$ case yields as a corollary an inequality due to Riesz and Sobolev,
and the $p=1$ case for the Shannon entropy is new and the most interesting from
an information-theoretic point of view.

In order to make the comparison with the classical Shannon-Stam entropy power inequality, we state the following
standard version of it \cite{CT06:book}, focusing on real-valued random variables
for simplicity.

\begin{thm}\label{thm:epi}\cite{Sha48, Sta59}
Let $X_1$ and $X_2$ be two independent $\RL$-valued random variables with finite differential entropies,
and finite variance.
Let $Z_1$ and $Z_2$ be two independent Gaussians such that
\begin{equation*}
h(X_i)=h(Z_i), \quad i=1,2.
\end{equation*}
Then
\ben
h(X_1+ X_2)\geq h(Z_1 + Z_2).
\een
\end{thm}
\vspace{.1in}

We also need the following lemmata, which we could not find explicitly stated in the literature. (The proofs are
not difficult, and given in later sections.)

\begin{lem}\label{lem:ent-pres}
If one of $h(X)$ and $h(X^*)$ is well defined, then so is the other one and we have
\ben
h(X)=h(X^*).
\een
\end{lem}
\vspace{.1in}

\begin{lem}\label{lem:var}
For any real random variable $X$,
\ben
\text{Var}(X^*)\leq\text{Var}(X).
\een
\end{lem}
\vspace{.1in}

First note that from Lemma~\ref{lem:ent-pres}, it follows that
\begin{equation*}
h(X^*_i)=h(X_i)=h(Z_i),\quad i=1,2.
\end{equation*}
Furthermore, if $X_1$ and $X_2$ have finite variance,
then Lemma~\ref{lem:var} implies that $X^*_1$ and $X^*_2$ have finite variance,
and therefore by the usual entropy power inequality (i.e., Theorem~\ref{thm:epi}), we have that
\be\label{eq:epicomp2}
h(X^*_1+ X^*_2)\geq h(Z_1+ Z_2).
\ee
On the other hand, the Main Theorem gives
\be\label{eq:epicomp1}
h(X_1+ X_2)\geq h(X^*_1+ X^*_2).
\ee
From \eqref{eq:epicomp1} and \eqref{eq:epicomp2}, we see that we have inserted the
quantity $h(X^*_1+ X^*_2)$ between the two sides of the entropy power inequality as stated in Theorem~\ref{thm:epi}.
In this sense, the $p=1$ and $k=2$ case is a kind of strengthening of Theorem~\ref{thm:epi}.

Let us outline how this note is organized.
Section~\ref{sec:rearr-basic} describes basic and well known facts about rearrangements
in a self-contained fashion.
Section~\ref{sec:rearr-strenth} discusses a result related to the Main Theorem but
for relative entropy (or more generally, R\'enyi divergence) rather than entropy.
In Section~\ref{sec:refine}, we
discuss connections of our result to a recent R\'enyi entropy power
inequality for independent random vectors
due to Bobkov and Chistyakov \cite{BC13:1}.

Then we give two related proofs of the Main Theorem, both of which are based
on the Rogers-Brascamp-Lieb-Luttinger inequality.
The first proof based on continuity considerations for R\'enyi entropy in the order is described in
Sections~\ref{sec:rearr-pre} and \ref{sec:pf}.
We include the first proof mainly because along the way, it clarifies various points that may be considered
folklore (in particular, the continuity of R\'enyi entropy in the order, which
has sometimes been taken for granted in the literature leading to incomplete statements of technical assumptions).

The second proof based on majorization ideas is simpler and more general, and described in Section~\ref{sec:major}.
Our approach here was inspired by slides 
of a 2008 talk of Bruce Hajek that we found online (after a draft of this paper was written with just the first proof).
Based on comments we received after the first draft of this paper was posted online,
it appears that the majorization-based approach to rearrangements is well known to experts
though there does not seem to be a clear exposition of it anywhere;
while its roots may be considered to lie implicitly in the famed text of Hardy, Littlewood and Polya \cite{HLP88:book}, it was significantly
taken forward in a difficult paper of Alvino, Trombetti and Lions \cite{ATL89}. As a result, a central technical result of this paper--
Theorem~\ref{thm:most-gen}-- may not be very surprising to experts. In fact,
after the first draft of this paper was circulated, it came to our attention that
when $\phi$ in Theorem~\ref{thm:most-gen} is non-negative, the $k=2$ case is Corollary 1 in Section 3.3
of Burchard's dissertation \cite{Bur94:phd}, where also the equality case is characterized
(this is much more difficult than the inequality itself). For non-negative $\phi$
and general $k$, Theorem~\ref{thm:most-gen} is proved in Corollary 3 in Section 3.4 of \cite{Bur94:phd}.

However, the main innovation in Theorem~\ref{thm:most-gen} is the extension to general convex functions
and the streamlined development using majorization that yields at one go a unified proof of
the Main Theorem for all $p$. It is pertinent to note that the most interesting case of the
Main Theorem, namely for Shannon differential entropy, would not follow from the earlier results. In the stated generality, Theorem~\ref{thm:most-gen} does not seem to have ever been written down before and its probabilistic implications-- including, in particular, the Main Theorem--
are developed for the first time in this work.

Section~\ref{sec:implicn} describes how various classical inequalities can be
seen as special cases of the Main Theorem,
while Section~\ref{sec:appln} discusses an application of the Main Theorem to bounding
the entropy of the sum of two independent uniform random vectors.

One application that rearrangement inequalities
have found in probability is in the area of isoperimetric inequalities for stochastic processes.
Representative works in this area include
Watanabe \cite{Wat83} on capacities associated to L\'evy processes, 
Burchard and Schmuckenschl\"ager \cite{BS01} on exit times of Brownian motions on the sphere or hyperbolic space,
Ba\~nuelos and M\'endez-Hern\'andez  on exit times and more for general L\'evy processes \cite{BM10},
and Drewitz, Sousi and Sun \cite{DSS13} on survival probabilities in a field of L\'evy-moving traps. 
Our results also have implications for stochastic processes and these are developed in Section~\ref{sec:levy}.

Finally, in Section~\ref{sec:pf-epi}, we give a new proof of the classical entropy power inequality
using the Main Theorem. This shares some features in common with the proof of Szarek and Voiculescu \cite{SV00},
but completely bypasses the use of Fisher information, MMSE or any differentiation of entropy functionals.

We have also obtained discrete analogues of several of the results of this paper; these 
analogues and their applications to combinatorics will be presented
elsewhere (see, e.g., \cite{WWM14:isit}).

For convenience of the reader, we collect here some (mostly standard) \textbf{notations} that will be used in the rest of the paper:
\begin{enumerate}
\item $\|x\|$: Euclidean norm of $x\in\mathbb{R}^n$.
\item $\{\text{Condition}\}$: equals $1$ if Condition is true and $0$ otherwise. For example, $\{f(x)>1\}=1$ if $f(x)>1$ and $0$ otherwise.
\item $\{x:\text{Condition}\}$: the set of $x$ such that Condition is true. For example, $\{x:f(x)>1\}$ is the set of all $x$ such that $f(x)>1$.
\item $\mathbb{I}_A(x)$: indicator function of the set $A$.
\item $t_+$:=$\max(t,0)$.
\item $f^+(x)$:=$\max(f(x),0)$.
\item $f^-(x)$:=$(-f)^+(x)$.
\item $f\star g$: the convolution of $f(x)$ and $g(x)$.
\item $\star_{1\leq i\leq n}f_i$: the convolution of the functions $f_i$, namely $f_1\star f_2\star\cdot\cdot\cdot\star f_n$.
\item $\phi'_+(t)$: right derivative of a function $\phi(t)$, defined on $\mathbb{R}$.
\end{enumerate}

\section{Basic facts about rearrangements}
\label{sec:rearr-basic}

We will try to make this section as self-contained as possible. For a Borel set $A$ with volume $|A|$, one can define its
spherically decreasing symmetric rearrangement $A^*$ by
\begin{equation*}
A^*=B(0,r),
\end{equation*}
where $B(0,r)$ stands for the open ball with radius $r$ centered at the origin and $r$ is determined by the condition that $B(0,r)$ has volume $|A|$. Here we use the convention that if $|A|=0$, then $A^*=\emptyset$ and that if $|A|=+\infty$, then $A^*=\mathbb{R}^n$.

\begin{figure}\label{fig}
\begin{center}
\includegraphics[height=2.4in]{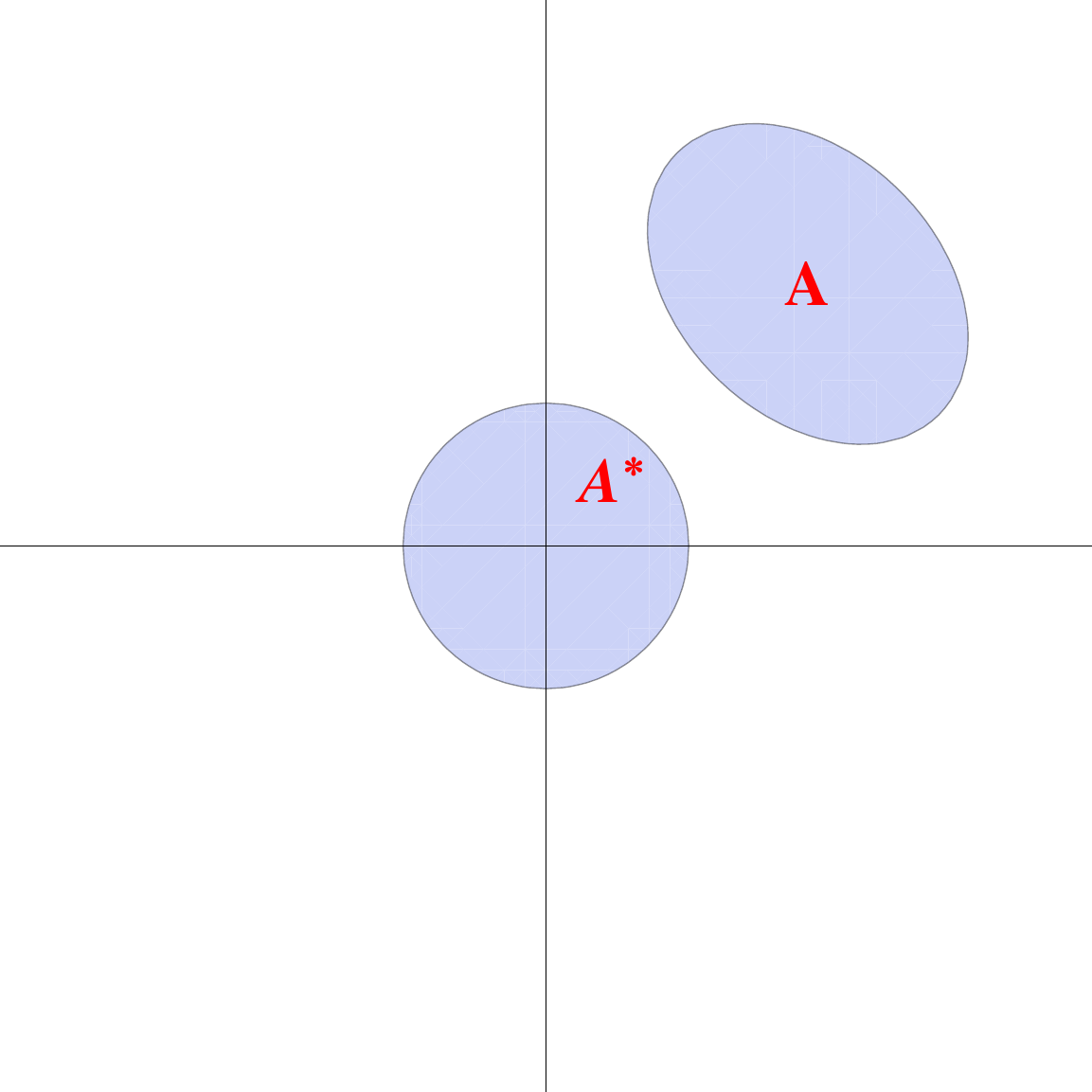}
\caption{Rearrangement of a set}
\end{center}
\end{figure}

Now for a measurable non-negative function $f$, we define its spherically decreasing symmetric rearrangement $f^*$ by:
\begin{equation*}
f^*(y)=\int_{0}^{+\infty}\{y\in B_t^*\}dt
\end{equation*}
where $B_t=\{x:f(x)>t\}$.

From the definition, it is intuitively clear that $\{x:f(x)> t\}^*=\{x:f^{*}(x)> t\}$ for all $t\geq0$. The proof of this is given in the following lemma, which is listed as an exercise in \cite{LL01:book}.

\begin{lem}\label{lem:basic}
$\{x:f(x)> t\}^*=\{x:f^{*}(x)> t\}$ for all $t\geq0$.
\end{lem}

\begin{proof}
Consider the function
\begin{equation*}
g(x,t)=\mathbb{I}_{\{y:f(y)>t\}^{*}}(x).
\end{equation*}
Observe that for fixed $x$, if for some $t_1$ we have
\begin{equation*}
g(x,t_1)=1,
\end{equation*}
then for all $t\leq t_1$, we would have
\begin{equation*}
g(x,t)=1.
\end{equation*}
Because of this,
\begin{equation*}
\{x:f^{*}(x)>t_1\}=\cup_{s>t_1}\{x:g(x,s)=1\}
\end{equation*}
\begin{equation*}
=\cup_{s>t_1}\{x:f(x)>s\}^{*}.
\end{equation*}

Notice that for $s_1<s_2$,
\begin{equation*}
\{x:f(x)>s_2\}^{*}\subseteq\{x:f(x)>s_1\}^{*}.
\end{equation*}
Hence
\begin{equation*}
\{x:f^{*}(x)>t_1\}=\bigcup_{n}\bigg\{x:f(x)>t_1+\frac{1}{n}\bigg\}^{*}.
\end{equation*}

Now observe that $\cup_{n}\{x:f(x)>t_1+\frac{1}{n}\}^{*}$ is an open ball with the same
Lebesgue measure as $\{x:f(x)>t_1\}^{*}$, which is also an open ball.
\end{proof}
\vspace{.1in}

\begin{rmk}
For any measurable subset $A$ of $[0,\infty)$, a generating class argument shows easily that $|\{x:f(x)\in A\}|=|\{x:f^*(x)\in A\}|$.
\end{rmk}

\begin{rmk}
A natural consequence of this is  that $f^*$ is lower semicontinuous. By the layer cake representation,
another consequence is that if $f$ is integrable, so is $f^*$ and $\|f\|_1=\|f^*\|_1$. In particular, $f^*$ is a probability density if $f$ is.
\end{rmk}

The second simple observation is that all R\'enyi entropies are preserved by rearrangements.

\begin{lem}\label{lem:preser-ent}
For any $p\in [0,1)\cup(1,\infty]$,
\ben
h_p(f)=h_p(f^*).
\een
For $p=1$, if one of $h(f)$ and $h(f^*)$ is well defined, then so is the other one and we have:
\ben
h(f)=h(f^*).
\een
\end{lem}
\vspace{.1in}

The preservation of $L^p$-norms (for $p>1$) by rearrangements is a very classical fact (see, e.g., \cite[Lemma 1.4]{Bur09:tut}),
although somewhat surprisingly the preservation of Shannon entropy does not seem to have explicitly noted anywhere in the
literature. We give a complete proof of a more general result, namely Lemma~\ref{lem:allpreser}, later.

Another useful fact is that composition with non-negative increasing functions on the left commutes with
taking rearrangements.

\begin{lem}\cite{LL01:book}\label{lem:increasing}
If $\Psi(x)$ is a non-negative real valued strictly increasing function defined on the non-negative real line, then
\ben
(\Psi(f))^*=\Psi(f^*)
\een
for any non-negative measurable function $f$.
\end{lem}

\begin{proof}
It suffices to show to show that the following two sets are equal for every $t$:
\begin{equation*}
\{x:(\Psi(f))^*>t\}=\{x:\Psi(f^*)>t\}.
\end{equation*}
By Lemma~\ref{lem:basic},
\begin{equation*}
\{x:\Psi(f^*)>t\}=\{x:f^*>\Psi^{-1}(t)\}=\{x:f>\Psi^{-1}(t)\}^{*}.
\end{equation*}
Again by Lemma~\ref{lem:basic},
\begin{equation*}
\{x:f>\Psi^{-1}(t)\}^{*}=\{x:\Psi(f)>(t)\}^{*}=\{x:(\Psi(f))^{*}>t\}.
\end{equation*}
\end{proof}
\vspace{.1in}

The final fact that will be useful later is that rearrangement decreases the $L_1$ distance between two functions. We refer to \cite{Bur09:tut} for a proof.

\begin{lem}\cite{Bur09:tut}\label{lem:L1}
Let $f$ and $g$ be two integrable non-negative functions. Then
\begin{equation*}
\|f^*-g^*\|_1\leq \|f-g\|_1.
\end{equation*}
\end{lem}

By construction, spherically symmetric decreasing rearrangements move the mass of
functions towards the origin. A fundamental rearrangement inequality expressing this
concentration is an inequality popularly known as the Brascamp-Lieb-Luttinger inequality,
which we state now.

\begin{thm}\label{thm:BLLfull}\cite{Rog57, BLL74}\,
For any measurable functions $f_i:\RL^n\ra [0,\infty)$, with $1\leq i \leq M$, and real numbers $a_{ij}, 1\leq i\leq M, 1\leq j\leq N$,
\ben\begin{split}
&\int_{\mathbb{R}^{nN}}\prod_{j=1}^N dx_j \prod_{i=1}^M f_i\bigg(\sum_{j=1}^N a_{ij}x_j\bigg)\\
&\leq\int_{\mathbb{R}^{nN}}\prod_{j=1}^N dx_j \prod_{i=1}^M f_i^*\bigg(\sum_{j=1}^N a_{ij}x_j\bigg).
\end{split}\een
\end{thm}

\begin{rmk}
To our considerable surprise, we found while preparing this paper that Theorem~\ref{thm:BLLfull}
was in fact discovered by C.~A.~Rogers \cite{Rog57} as far back as 1957, and
using a similar proof idea as Brascamp, Lieb and Luttinger \cite{BLL74}
rediscovered in 1974. This historical fact
does not seem to be widely known, but it is the reason we call Theorem~\ref{thm:BLLfull}
the Rogers-Brascamp-Lieb-Luttinger inequality elsewhere in this paper.
\end{rmk}

\begin{rmk}
As noted in \cite{BLL74}, Theorem~\ref{thm:BLLfull} is only nontrivial when $M>N$.
\end{rmk}

In fact, we only need the following special but important case of the Rogers-Brascamp-Lieb-Luttinger inequality in this paper.

\begin{thm}\label{thm:BLL}
For any non-negative measurable functions $f_i, 1\leq i \leq k$, on $\RL^n$, we have
\ben
\int f_1(y) \bigg[\star_{2\leq i\leq k}f_i(y) \bigg] dy
\een
\ben
\leq\int f_1^*(y) \bigg[\star_{2\leq i\leq k}f_i^*(y) \bigg] dy.
\een
\end{thm}

\begin{proof}
By definition, we have
\ben
\int f_1(y)\bigg[\star_{2\leq i\leq k}f_i(y) \bigg] dy
\een
\begin{equation*}
=\int_{\mathbb{R}^{n(k-1)}}\prod_{j=1}^{k-1} dx_j \prod_{i=1}^k f_i\bigg(\sum_{j=1}^{k-1}a_{ij}x_j\bigg),
\end{equation*}
where $a_{1j}=\delta_{j,1}$, $a_{2j}=-1+2\delta_{j,1}$ and $a_{ij}=\delta_{i-1,j},i>2$.
Hence we can apply Theorem~\ref{thm:BLLfull} with $N=k-1$, $M=k$ to conclude.
\end{proof}
\vspace{.1in}

\begin{rmk}
For $k=1$, Theorem~\ref{thm:BLL} is called the Hardy-Littlewood inequality. For $k=2$, it is called the Riesz
or Riesz-Sobolev inequality.
(Riesz \cite{Rie30} actually proved only the one dimensional case, but it was generalized by Sobolev \cite{Sob38}
to higher dimensions. See, for example, \cite{Bur09:tut} for historical perspective on all these inequalities.)
For $k>2$, as demonstrated during the proof, Theorem~\ref{thm:BLL} is a special case of Theorem~\ref{thm:BLLfull}.
\end{rmk}

\begin{rmk}\label{rmk:bll}
Observe that when $f_2, \ldots, f_k$ are densities, we may interpret Theorem~\ref{thm:BLL} probabilistically
as follows. Let $X_1, \ldots, X_M$ be random vectors with densities on $\RL^n$. Then for any
non-negative measurable function $u$ on $\RL^n$,
\ben\begin{split}
&\mathbb{E} u\bigg(\sum_{i=1}^M X_i\bigg)
&\leq \mathbb{E} u^* \bigg(\sum_{i=1}^M X_i^* \bigg) .
\end{split}\een
\end{rmk}

In the following, we will see that our Main Theorem is a consequence of Theorem~\ref{thm:BLL}
(and in fact, they are mathematically equivalent).


\section{Moment and Relative Entropy Inequalities for Rearrangements}
\label{sec:rearr-strenth}

In this section, we will show some moment and relative entropy inequalities, which are useful later.
\begin{lem}\label{lem:finvar}
Let $g(x)$ be a non-negative increasing function defined on the non-negative real line and $f_i, 1\leq i\leq k$ be probability densities.
Then
\ben\begin{split}
&\mathbb{E}g(\|X_1+X_2+\cdot\cdot+X_k\|) \\
&\geq \mathbb{E}g(\|X_1^*+X_2^*+\cdot\cdot+X_k^*\|),
\end{split}\een
where all random vectors are independent, $X_i$ is distributed according to $f_i$, and $X_i^*$ is distributed according to $f_i^*$.
\end{lem}

\begin{proof}
For any $t>0$, we can apply Theorem~\ref{thm:BLL} (as interpreted in Remark~\ref{rmk:bll}) to obtain
\ben
\mathbb{E}\bigg[e^{-tg(\|X_1+X_2+\cdot\cdot+X_k\|)}\bigg]\leq\mathbb{E}\bigg[e^{-tg(\|X_1^*+X_2^*+\cdot\cdot+X_k^*\|)}\bigg],
\een
since $\big(e^{-tg(\|x\|)}\big)^*=e^{-tg(\|x\|)}$ almost everywhere.
Hence we get:
\be\label{eq:localvariance1}\begin{split}
\mathbb{E}\bigg[1-e^{-tg(\|X_1+X_2+\cdot\cdot+X_k\|)}\bigg]
\geq\mathbb{E}\bigg[1-e^{-tg(\|X_1^*+X_2^*+\cdot\cdot+X_k^*\|)}\bigg].
\end{split}\ee
Now note that $\frac{1-e^{-tg(\|x\|)}}{t}$ is monotonically increasing to $g(\|x\|)$ as $t$ goes to zero, and that
\be\label{eq:localvariance2}
\frac{1-e^{-tg(\|x\|)}}{t}\leq g(\|x\|)
\ee
for any $t>0$. Hence, as $t$ goes to $0$, we can apply monotone convergence on the right side of \eqref{eq:localvariance1},
and use the inequality \eqref{eq:localvariance2} on the left side of \eqref{eq:localvariance1} to obtain the claimed result.
\end{proof}
\vspace{.1in}

It is easy to see that the density of $(X-c)^*$ is the same as that of $X^*$. Hence a simple consequence of Lemma~\ref{lem:finvar} is:

\newcommand{\cov}{{\rm Cov}}
\begin{cor}
For any random vector $X$ with finite covariance matrix,
$\mathbb{E}\|X^*\|^2\leq\mathbb{E}\|X-\mathbb{E}[X]\|^2$.
Equivalently, $\tr \cov (X^*)\leq \tr\cov(X)$.
\end{cor}

Lemma~\ref{lem:finvar}, together with Lemma~\ref{lem:preser-ent}, immediately imply that $D(f^*\|g)\leq D(f\|g)$ where
$g$ is a non-degenerate isotropic Gaussian. In fact, we have the following more general fact, easiest to state
in terms of the R\'enyi divergence (see \cite{EH12} for a recent survey on this). Recall that the
R\'enyi divergence of order $\alpha$ between any two densities $f$ and $g$ is defined as
\ben
D_{\alpha}(f\|g)=\frac{1}{\alpha-1}\int f^{\alpha}g^{1-\alpha}dx
\een
for $\alpha\in (0,1)\cup(1,\infty)$, and that  $D_1(f\|g)=D(f\|g)$ is simply the usual relative entropy
between $f$ and $g$.

\begin{prop}
Let $f$ and $g$ be two probability densities. Then
\begin{equation*}
D_{\alpha}(f^*\|g^*)\leq D_{\alpha}(f\|g),
\end{equation*}
where $0<\alpha\leq 1$.
\end{prop}
\begin{proof}
By Lemma~\ref{lem:increasing}, we have:
\ben
\int (f^*)^{\alpha}(g^*)^{1-\alpha}dx=\int (f^{\alpha})^*(g^{1-\alpha})^*dx,
\een
where $0<\alpha<1$. Now an easy application of Theorem~\ref{thm:BLL} will lead to:
\be\label{eq:reldec}
\int (f^*)^{\alpha}(g^*)^{1-\alpha}dx\geq\int f^{\alpha}g^{1-\alpha}dx,
\ee
which is equivalent to $D_{\alpha}(f^*\|g^*)\leq D_{\alpha}(f\|g)$.

It remains to prove the statement for relative entropy.
We first do this under the assumptions that $f\ll g$, $f^*\ll g^*$, and
that the respective likelihood ratios $r(x)$ and $r_*(x)$ are uniformly bounded.
By these assumptions, we can rewrite \eqref{eq:reldec} in terms of likelihood ratios:
\ben
\int (r_*)^{\alpha}g^*dx\geq\int r^{\alpha}g dx.
\een
Noting $\int r_*(x)g^*(x)dx=\int r(x)g(x)dx=1$, we get:
\be\label{eq:localrel} \int\frac{(r_*)^{\alpha}-r_*}{\alpha-1}g^*dx\leq\int\frac{r^{\alpha}-r}{\alpha-1}g dx.
\ee
Since we have assumed the uniform boundedness of both $r(x)$ and $r_*(x)$, as $\alpha$ goes to $1$, we can apply bounded convergence to both sides of \eqref{eq:localrel} (we omit the details) to obtain:
\ben
\int r_*\log(r_*)g^*dx\leq\int r\log(r)g dx,
\een
which is equivalent to
\ben
D(f^*\|g^*)\leq D(f\|g).
\een
In the general case, by what has already been proved, we have:
\ben
D(f^*\|(\lambda f+(1-\lambda)g)^*)\leq D(f\|\lambda f+(1-\lambda)g),
\een
where $0<\lambda<1$. Here we have used the fact $(\lambda f+(1-\lambda)g)^*\geq\lambda f^*$ since the $*$ operation is order preserving \cite{Bur09:tut}.
Note that by Lemma~\ref{lem:L1},
\ben\begin{split}
\|g^*-[\lambda f+(1-\lambda)g]^*\|_1
&\leq\|g-[\lambda f+(1-\lambda)g]\|_1 \\
&\leq2\lambda.
\end{split}\een
Hence as $\lambda$ goes to $0$, $(\lambda f+(1-\lambda)g)^*$ converges to $g^*$ in total variation distance. But it
is well known that the relative entropy functional is jointly lower-semicontinuous with respect to the
topology of weak convergence (see, e.g., \cite{DZ98:book}). Hence:
\be\label{eq:localrel1}
D(f^*\|g^*)\leq\liminf_{\lambda\downarrow0}D(f^*\|(\lambda f+(1-\lambda)g)^*).
\ee
On the other hand, the relative entropy functional is jointly convex \cite{DZ98:book}. We get:
\be\label{eq:localrel2}
D(f\|\lambda f+(1-\lambda)g)\leq(1-\lambda)D(f\|g).
\ee
Combining \eqref{eq:localrel1} and \eqref{eq:localrel2}, we can conclude.
\end{proof}


\section{Towards the optimal R\'enyi EPI}
\label{sec:refine}

Very recently, Bobkov and Chistyakov \cite{BC13:1} obtained a generalization of the
entropy power inequality (henceforth EPI), due to Shannon \cite{Sha48} and Stam \cite{Sta59},
to R\'enyi entropy of any order $p> 1$. First we recall that the {\it R\'enyi entropy power}
of order $p$ is defined for $\RL^n$-valued random vectors by
\ben
N_p(X)=\exp\bigg\{\frac{2h_p(X)}{n}\bigg\}.
\een

\begin{thm}\label{thm:bc}\cite{BC13:1}
If $X_1, \ldots, X_k$ are independent random vectors taking values in $\RL^n$, then for any $p\geq 1$,
\ben
N_p(X_1+\ldots+X_k) \geq c_p \sum_{i=1}^k N_p(X_i) ,
\een
where $c_p$ is a constant depending only on $p$. Moreover, one may take $c_1=1$,
\ben
c_p= \frac{1}{e} p^{\frac{1}{p-1}}  \quad \text{for } p\in (1, \infty) ,
\een
and $c_\infty=\frac{1}{e}$. If $n=1$, we may take $c_\infty=\frac{1}{2}$.
\end{thm}

Note that the $p=1$ case is simply the Shannon-Stam EPI, which is sharp in any dimension with
equality only for Gaussians with proportional covariance matrices.
In general, however, Theorem~\ref{thm:bc} is not sharp, and \cite{BC13:1} does not suggest what the
optimal constants or the extremizing distributions might be.

\begin{rmk}\label{rmk:bc-inf}
For any dimension $n$, restricting to $k=2$ and $p=\infty$,
it turns out that the sharp constant is $\half$, which is achieved
for two identical uniform distributions on the unit cube; this is observed in another paper of Bobkov and Chistyakov \cite{BC12}.
However, note that the maximum of the density of convolution of two uniforms on the ball with unit volume (not unit ball)
is $1$, while the maximum of the density of each of them is also $1$. Hence, some non-uniqueness of the
extremal distribution arises (at least for $p=\infty$) .
Indeed, for $k=2$ and $p=\infty$, uniform distributions on any symmetric convex set $K$ (i.e., $K$ is convex,
and $x\in K$ if and only if $-x\in K$) 
of volume $1$  will be extremal: if $X$ and $X'$ are independently distributed according to $f=\mathbb{I}_K$,
then denoting the density of $X-X'$ by $u$, we have
\ben
\|u\|_{\infty}=u(0)=\int f^2(x) dx = 1 = \|f\|_{\infty} ,
\een
so that $N_{\infty}(X+X')=N_{\infty}(X-X')= N_{\infty}(X)=\half [N_{\infty}(X)+N_{\infty}(X')]$.
\end{rmk}
\vspace{.1in}

Our Main Theorem may be seen as refining Theorem~\ref{thm:bc} (in a similar
way to how the $p=1$ case of it refined the classical EPI). In order to do this,
however, we need to recast Theorem~\ref{thm:bc} in a different, more precise, conjectural form, which
suggests the extremizing distributions of Theorem~\ref{thm:bc}.

When dealing with optimization problems involving R\'enyi entropies, it is quite common for
a certain class of generalized Gaussians to arise. A rich collection of such generalized Gaussians
has been studied in the literature. The ones that are of interest to us are a one-parameter family
of distributions, indexed by a parameter $-\infty<\beta\leq \frac{2}{n+2}$, of the following form:
$g_0$ is the standard Gaussian density in $\RL^n$, and for $\beta\neq 0$,
\ben
g_\beta(x)= A_{\beta} \bigg(1-\frac{\beta}{2}\|x\|^2 \bigg)_{+}^{\frac{1}{\beta}-\frac{n}{2}-1} ,
\een
where $A_{\beta}$ is a normalizing constant
(which can be written explicitly in terms of gamma functions if needed).
We call $g_\beta$ the {\it standard generalized Gaussian} of order $\beta$;
any affine function of a standard generalized Gaussian yields a ``generalized Gaussian''.
Observe that the densities $g_\beta$ (apart from the obviously special value $\beta=0$)
are easily classified into 
two distinct ranges where they behave differently. {\it First}, for $\beta<0$,
the density is proportional to a negative power of $(1+b\|x\|^2)$ for a positive constant $b$,
and therefore correspond to measures with full support on $\RL^n$ that are heavy-tailed.
For $\beta>0$, note that $(1-b\|x\|^2)_+$ with positive $b$ is non-zero only for $\|x\|<b^{-\half}$, and is concave in this region.
Thus any density in the {\it second} class, corresponding to 
$0<\beta\leq\frac{2}{n+2}$,
is a positive power of $(1-b\|x\|^2)_+$, and is thus a concave function supported on a centered Euclidean ball
of finite radius.
In particular, note that $g_{\frac{2}{n+2}}$ is the uniform distribution on the
Euclidean ball of radius $\sqrt{n+2}$. It is pertinent to note that although the first class includes many
distributions from what one might call the ``Cauchy family'', it excludes the standard Cauchy distribution;
indeed, not only do all the generalized Gaussians defined above have finite variance, but in fact the
form has been chosen so that, for $Z\sim g_\beta$,
\ben
\E [\|Z\|^2 ]=n
\een
for any $\beta$.
Incidentally, the generalized Gaussians are called by many different names in the literature;
the one other nomenclature that is perhaps worth noting for its relevance in statistics is
that the $\beta<0$ class is also called the Student-$r$ class, while the
$0<\beta\leq \frac{2}{n+2}$ class is also called the Student-$t$ class.

For $p>\frac{n}{n+2}$, define $\beta_p$ by
\ben
\frac{1}{\beta_p}= \frac{1}{p-1}+\frac{n+2}{2} ;
\een
note that $\beta_p$ ranges from $-\infty$ to $\frac{2}{n+2}$ as $p$ ranges from
$\frac{n}{n+2}$ to $\infty$. Henceforth we will write $Z^{(p)}$ for a random vector
drawn from $g_{\beta_p}$.

Costa, Hero and Vignat \cite{CHV03} showed that the maximizers of R\'enyi entropy
under covariance constraints are the generalized Gaussians; this fact was later obtained by
Lutwak, Yang and  Zhang \cite{LYZ07a} in a more general setting as what they called
``moment-entropy inequalities''. We find the following formulation convenient.

\begin{thm}\label{thm:lyz}\cite{CHV03, LYZ07a}
If $X$ is a random vector taking values in $\RL^n$, then for any $p>\frac{n}{n+2}$,
\ben
\frac{\E[\|X\|^2]}{N_p(X)} \geq \frac{\E[\|Z^{(p)}\|^2]}{N_p(Z^{(p)})}.
\een
\end{thm}

Clearly  Theorem~\ref{thm:lyz} implies that under a variance constraint,
the R\'enyi entropy power of order $p$ is maximized by the generalized Gaussian $Z^{(p)}$.

This leads us to the following conjecture, refining Theorem~\ref{thm:bc}.

\begin{conj}\label{conj:renyi-epi}
Let $X_1, \ldots, X_k$ be independent random vectors taking values in $\RL^n$,
and $p> \frac{n}{n+2}$. Suppose $Z_i$ are independent random vectors, each a scaled version of $Z^{(p)}$.
such that $h_p(X_i)=h_p(Z_i)$. Then
\ben
N_p(X_1+\ldots+X_k) \geq N_p(Z_1+\ldots+Z_k) .
\een
\end{conj}
This conjecture is true for at least three important special cases:
\begin{enumerate}
\item $p=1,\forall n,\forall k$: This is the classical EPI;
\item $p=\infty,k=2,\forall n$: This follows from a recent result of Bobkov and Chistyakov \cite{BC12} as explained in Remark~\ref{rmk:bc-inf};
\item $p=\infty,n=1,\forall k$: This is a relatively old but nontrivial result of Rogozin \cite{Rog87:1}. 
\end{enumerate}
In principle, Conjecture~\ref{conj:renyi-epi} suggests optimal constants for Theorem~\ref{thm:bc}; they should
simply be those that arise in comparing
$N_p(Z_1+\ldots+Z_k)$ with $\sum_{i=1}^k N_p(Z_i)$ (which,
by construction is simply $\sum_{i=1}^k N_p(X_i)$).
The optimal constants would depend on $k$, $p$ and $n$.
This explains the precise reason why the optimal constant is 1
in the classical case $p=1$; it is because sums of independent Gaussians
are independent Gaussians.
For the case $k=2$, we make a more aggressive conjecture.
We would like to know
\ben
\inf_{a_1, a_2} \frac{N_p(a_1 Z^{(p)}_1+ a_2 Z^{(p)}_2)}{N_p(a_1 Z^{(p)}_1)+N_p(a_2 Z^{(p)}_2)},
\een
where $Z^{(p)}_1$ and $Z^{(p)}_2$ are independently drawn from $g_{\beta_p}$.
Although we do not have a rigorous argument, symmetry seems to suggest that $a_1=a_2$ would be optimal here,
and this is borne out by some numerical tests in one dimension.
One can compute (see, e.g., \cite{LLYZ13}) that, for $p\neq 1$,
\ben
N_p(Z^{(p)}) = A_{\beta}^{-\frac{2}{n}} \big(1-\half n\beta_p\big)^{\frac{2}{n(1-p)}} .
\een
(For $p=1$, $N_p(Z^{(p)})$ is just the Shannon entropy power of the standard Gaussian,
which is well known to be $2\pi e$.) It appears to be much harder to compute $N_p$ for the self-convolution
of $Z^{(p)}$.

\begin{conj}\label{conj:renyi-epi-2}
Let $X_1$ and $X_2$ be independent random vectors taking values in $\RL^n$,
and $p> \frac{n}{n+2}$. Then
\be\label{eq:conj}
N_p(X_1+X_2) \geq C_{p,n} [N_p(X_1)+ N_p(X_2)] ,
\ee
where $C_{p,n}=\frac{1}{2}\frac{N_p(Z^{(p)}_1+Z^{(p)}_2)}{N_p(Z^{(p)})}$
and $Z^{(p)}_1, Z^{(p)}_2$ are independently drawn from $g_{\beta_p}$.
\end{conj}
This conjecture is true in two known cases: when $p=1$, $C_{1,n}=1$, and it is the classical EPI;
when $p=\infty$, $C_{\infty,n}=\frac{1}{2}$, and how this follows from
the work of Bobkov and Chistyakov \cite{BC12} is explained in Remark~\ref{rmk:bc-inf}.

We remark that in terms of importance, Conjecture~\ref{conj:renyi-epi} is far more important.
If we can prove Conjecture~\ref{conj:renyi-epi} for $k=2$, then proving or disproving
Conjecture~\ref{conj:renyi-epi-2} reduces in principle to a calculus problem. Also, we mention that maybe
the right formulation to start with is Conjecture~\ref{conj:renyi-epi}, since $C_{p,n}$ might not be
analytically computable or may have a complicated expression. For example, when $n=1, p=2$, we computed $C_{2,1}$ using Mathematica:
\ben
C_{2,1}=\frac{166753125}{16 \left(\frac{573635 \sqrt{\frac{5}{2}}}{2}-142365 \sqrt{10}\right)^2}\doteq0.956668.
\een

We mention for completeness that Johnson and Vignat \cite{JV07} also demonstrated what they
call an ``entropy power inequality for R\'enyi entropy'', for any order $p\geq 1$. However, their
inequality does not pertain to the usual convolution, but a new and somewhat complicated
convolution operation (depending on $p$). This new operation reduces to the usual convolution for $p=1$,
and has the nice property that the convolution of affine transforms of independent copies
of $Z^{(p)}$ is an  affine transform of $Z^{(p)}$ (which, as observed above, fails for the usual convolution).

A variant of the classical entropy power inequality is the
concavity of entropy power of a sum when one of the summands is Gaussian,
due to Costa \cite{Cos85b}.
Savar\'e and Toscani  \cite{ST14} recently proposed a generalization of Costa's
result to R\'enyi entropy power, but the notion of concavity they use based on
solutions of a nonlinear heat equation does not have obvious probabilistic meaning,
and their work also does not seem directly connected to the approach discussed in this
section. The definition of R\'enyi entropy power used in \cite{ST14}
has a different constant in the exponent ($\frac{2}{n}+p-1$ as opposed to $\frac{2}{n}$),
and it is conceivable, as a reviewer suggested, that Conjecture~\ref{conj:renyi-epi-2} 
is true for $p\in (1,\infty)$ only with this modified definition of $N_p$.


\section{Preliminaries for the First Proof}
\label{sec:rearr-pre}

We first state a lemma, which seems to be folklore in the continuous case \cite{DCT91} and allows us to to obtain the cases when $p=0,1,\infty$ as limiting cases.
\begin{lem}\label{lem:key}
\begin{enumerate}
\item[(i)] The following limit is well defined:
\ben
h_0(f)=\lim_{p\rightarrow0^{+}}h_p(f)=\log|supp(f)|,
\een
where $supp(f)$ is the support of $f$, defined as the set $\{x: f(x)\neq0\}$.

\item[(ii)]  The following two limits are well defined:
\ben
h^{+}_1(f)=\lim_{p\rightarrow1^{+}}h_p(f),
\een
\ben
h^{-}_1(f)=\lim_{p\rightarrow1^{-}}h_p(f).
\een
If $h^{+}_1(f)>-\infty$, then $h(f)$ is well defined (possibly $+\infty$) and we have:
\ben
h^{+}_1(f)=h(f).
\een
If $h^{-}_1(f)<+\infty$, then $h(f)$ is well defined (possibly $-\infty$) and we have:
\ben
h^{-}_1(f)=h(f).
\een

\item[(iii)]  The following limit is well defined:
\ben
h_{\infty}(f)=\lim_{p\rightarrow+\infty}h_p(f),
\een
and we have:
\ben
h_{\infty}(f)=-\log\|f\|_{\infty},
\een
where $\|f\|_{\infty}$ is the essential supremum of $f$.

If we suppose that $f$ is defined everywhere and lower semicontinuous, then
\ben
h_{\infty}(f)=-\log\sup_{x}f(x).
\een

\item[(iv)]  Suppose $h(f)$ is well defined (possibly $\infty$) and let $q\in(0,1),p\in(1,+\infty)$. Then
\ben
h_q(f)\geq h(f)\geq h_p(f).
\een
\end{enumerate}
\end{lem}

\begin{proof}
We first show that all four limits are well defined. By Lyapunov's inequality, the function $g$, defined by
\begin{equation*}
g(p)=\log\int f^p(x)dx
\end{equation*}
is a convex function from $(0,+\infty)$ to $(-\infty,+\infty]$. Since $f$ is assumed to be probability density, we have $g(1)=0$.

Consider $p,q\in(0,1)$ with $p<q$. Convexity will give us:
\begin{equation*}
\frac{g(1)-g(p)}{1-p}\leq\frac{g(1)-g(q)}{1-q},
\end{equation*}
which is equivalent to $h_p(f)\geq h_q(f)$. This monotonicity guarantees that the $h_0(f)$ and $h^{-}_1(f)$ exist, possibly as extended real numbers.

Similarly, if we consider $p,q\in(1,\infty)$ with $p<q$, convexity will give us:
\begin{equation*}
\frac{g(p)-g(1)}{p-1}\leq\frac{g(q)-g(1)}{q-1},
\end{equation*}
which is equivalent to $h_p(f)\geq h_q(f)$. This will guarantee the existence of the other two limits, possibly as extended real numbers.

Now we prove the lemma.

\begin{enumerate}
\item[(i)]
We first assume $|supp(f)|<\infty$. It's clear that we only need to show:
\ben
\lim_{p\rightarrow0^{+}}\int f^p(x)dx=|supp(f)|.
\een
Note that
\ben
\int f^p(x)dx=\int\{f(x)\neq0\}f^p(x)dx
\een
and
\ben
\{f(x)\neq0\}f^p(x)\leq\max(f(x),1)\{f(x)\neq0\}.
\een
But $\max(f(x),1)\{f(x)\neq0\}$ is integrable by our assumption that $|supp(f)|<\infty$. Hence by dominated convergence, we obtain:
\ben
\lim_{p\rightarrow0^{+}}\int f^p(x)dx=\lim_{p\rightarrow0^{+}}\int\{f(x)\neq0\}f^p(x)dx
=\int\{f(x)\neq0\}dx=|supp(f)|.
\een
On the other hand, if $|supp(f)|=+\infty$, then
\ben
\lim_n \bigg|\bigg\{x:f(x)\geq\frac{1}{n}\bigg\}\bigg|=+\infty.
\een
Note that
\ben
\lim_{p\rightarrow0^{+}}\int f^p(x) dx
\geq \lim_{p\rightarrow0^{+}} \frac{1}{n^p} \bigg|\bigg\{x:f(x)\geq\frac{1}{n}\bigg\}\bigg|
=\bigg|\bigg\{x:f(x)\geq\frac{1}{n}\bigg\}\bigg| .
\een
Hence we have
\ben
\lim_{p\rightarrow0^{+}}\int f^p(x)dx=+\infty=|supp(f)|.
\een


\item[(ii)]
Clearly, assuming $h^{+}_1(f)>-\infty$ implies the existence of $\epsilon>0$, such that:
\ben
\int f^{1+\epsilon}(x)dx<+\infty.
\een
We first show that the function $g(p)$ is continuous at $p=1$. Take $1<p<1+\epsilon$. Then
\ben
f^p(x)\leq\{f(x)\leq1\}f(x)+\{f(x)>1\}f^{1+\epsilon}(x).
\een
But $\{f(x)\leq1\}f(x)+\{f(x)>1\}f^{1+\epsilon}(x)$ is integrable, hence we can apply dominated convergence to conclude that $g(p)$ is continuous at $p=1$.

Next we show that for $1<p<1+\epsilon$, $g(p)$ is differentiable, with a finite derivative given by:
\ben
g'(p)=\int f^p\log(f)dx.
\een
This basically  follows from dominated convergence and the details of the argument can be found in Lemma 6.1 in \cite{BMW11}. We do not repeat it here.

Now, we show that the limit
\ben
-\lim_{p\rightarrow1^{+}}g'(p)
\een
exists (possibly $+\infty$) and equals $h(f)$. Note that
\ben
\int f^p\log(f)dx=\int_{\{x:f(x)>1\}}f^p(x)\log^{+}(f(x))dx
-\int_{\{x:0<f(x)<1\}}f^p(x)\log^{-}(f(x))dx.
\een
For the first term, we can apply dominated convergence using $\int f^{1+\epsilon}dx<+\infty$. While for the second term, we can apply monotone convergence.

Finally, we can apply a version of L'Hospital's rule to conclude that
\ben
\lim_{p\rightarrow1^+}\frac{g(p)}{1-p}=-\lim_{p\rightarrow1^+}g'(p)=h(f).
\een

A similar analysis applies to the assumption $h^{-}_1(f)<+\infty$.

\item[(iii)]
The statement $h_\infty(f)=-\log\|f\|_\infty$ follows from a classical analysis fact \cite{Rud87:book}.

For the other statement, one direction is easy:
\begin{equation*}
h_p(f)=\frac{1}{1-p}\log\int f^p(x)dx
\end{equation*}
\begin{equation*}
\geq\frac{1}{1-p}\log\bigg(\sup_xf^{p-1}(x)\bigg)=-\log\sup_xf(x).
\end{equation*}
Note that the above is automatically true if $\sup_xf(x)=+\infty$. Now fix $y\in supp(f)$ and $0<\delta<f(y)$. By lower semicontinuity, we can find an open ball $B(y,r)$, centered at $y$ with radius $r>0$, such that:
\begin{equation*}
\inf_{x\in B(y,r)}f(x)\geq f(y)-\delta.
\end{equation*}
From this, we easily obtain:
\begin{equation*}
h_p(f)\leq\frac{\log\bigg((f(y)-\delta)^p|B(y,r)|\bigg)}{1-p}.
\end{equation*}
By first sending $p$ to infinity and then sending $\delta$ to zero, we arrive at:
\begin{equation*}
h_\infty(f)\leq-\log f(y).
\end{equation*}
Taking infimum of the right hand side over $y\in supp(f)$, we conclude:
\begin{equation*}
h_\infty(f)\leq-\log\sup_{y\in supp(f)}f(y)=-\log\sup_{y}f(y).
\end{equation*}

\item[(iv)]
Since $h(f)$ is assumed to be well defined, we can apply Jensen's inequality to obtain:
\ben
\log\int f^p(x)dx=\log\int f^{p-1}(x)f(x)dx
\geq(p-1)\int\log(f(x))f(x)dx.
\een
This will give $h(f)\leq h_p(f)$. Similarly, we can obtain the other inequality.
\end{enumerate}
\end{proof}
\vspace{.1in}

\begin{rmk}
As was observed in \cite{BM11:it}, 
there exist densities such that $h_p(X)=-\infty$
for every $p>1$ but $h(X)$ is well defined and finite. An example of such a density is
\ben
f(x)=\frac{c}{x\log^3(1/x)} , \,\, 0<x<\half ,
\een
where $c$ is a normalizing constant. This shows that the continuity of R\'enyi entropy
in the order at $p=1$ is not automatic, and means that one has to be cautious about
the generality in which proofs based on this continuity apply (e.g., the proof
of the entropy power inequality from Young's inequality with sharp constant in \cite{DCT91}
needs an additional condition as in Lemma~\ref{lem:key}(ii)).
\end{rmk}

For later reference, we also record the following corollary, which is clear from the proof of Lemma~\ref{lem:key}(ii).

\begin{cor}\label{cor:easy}
The fact that $h^{+}_1(f)>-\infty$ is equivalent to the existence of $\epsilon>0$, such that:
\begin{equation*}
\int f^{1+\epsilon}(x)dx<+\infty.
\end{equation*}
Similarly, the fact that $h^{-}_1(f)<+\infty$ is equivalent to the existence of $1>\epsilon>0$, such that:
\begin{equation*}
\int f^{1-\epsilon}(x)dx<+\infty.
\end{equation*}
\end{cor}


\section{First Proof of Main Theorem}
\label{sec:pf}

\begin{thm}\label{thm:main1}
If $f_i$ are densities on $\RL^n$ and $f_i^*$ are their spherically symmetric rearrangements, then
\ben
h_p(f_1\star f_2\star\cdot\cdot\star f_k) \geq h_p(f_1^* \star f_2^*\star\cdot\cdot f_k^*) ,
\een
for any $p\in[0,1)\cup(1,\infty]$. For $p=1$, if $h(f_1\star f_2\star\cdot\cdot f_k)$ is well defined, then
\ben
h(f_1\star f_2\star\cdot\cdot\star f_k)\geq h^{+}_1(f_1^* \star f_2^*\star\cdot\cdot f_k^*);
\een
if $h(f_1^* \star f_2^*\star\cdot\cdot f_k^*)$ is  well defined, then
\ben
h^{-}_1(f_1\star f_2\star\cdot\cdot\star f_k) \geq h(f_1^* \star f_2^*\star\cdot\cdot f_k^*).
\een
In particular, if one of the densities, say $f_1$, satisfies
\ben
\int f^{1+\epsilon}_1(x)dx<+\infty,
\een
for some $\epsilon>0$, then both $h(f_1^* \star f_2^*\star\cdot\cdot f_k^*)$ and
$h(f_1\star f_2\star\cdot\cdot f_k)$ are well defined and we have:
\ben
h(f_1\star f_2\star\cdot\cdot\star f_k)\geq h(f_1^* \star f_2^*\star\cdot\cdot f_k^*).
\een
\end{thm}

\begin{proof}
\noindent {\bf Case 1: $p\in(1,+\infty)$}.\\
By definition of R\'enyi entropy, it suffices to show
\be\label{eq:lem-f}
\|f_1\star f_2\star\cdot\cdot f_k\|_p\leq\|f_1^* \star f_2^*\star\cdot\cdot f_k^*\|_p.
\ee
Note that by duality,
\begin{equation*}
\|f_1\star f_2\star\cdot\cdot f_k\|_p=\sup_{\|g\|_q=1}\int_{\mathbb{R}^n}g(x)f_1\star f_2\star\cdot\cdot f_k(x)dx,
\end{equation*}
where $\frac{1}{p}+\frac{1}{q}=1$. Hence we can apply Theorem~\ref{thm:BLL} to obtain:
\ben
\int_{\mathbb{R}^n}g(x)f_1\star f_2\star\cdot\cdot f_kdx
\leq\int_{\mathbb{R}^n}|g|^*(x)f_1^* \star f_2^*\star\cdot\cdot f_k^*dx.
\een
Due to Lemma~\ref{lem:preser-ent}, $\||g|^*\|_q=1$ and again by duality,
\ben
\int_{\mathbb{R}^n}|g|^*(x)f_1^* \star f_2^*\star\cdot\cdot f_k^*dx
\leq\|f_1^* \star f_2^*\star\cdot\cdot f_k^*\|_p.
\een
Hence inequality \eqref{eq:lem-f} is shown.

\noindent {\bf Case 2: $p=\infty$}.\\
This follows from Case 1 and Lemma~\ref{lem:key}.

\noindent {\bf Case 3: $p\in(0,1)$}.\\
As mentioned before, the $k=2$ case was proved by \cite[Proposition 9]{BL76b}. It is straightforward to extend
the argument there to general $k$. We give the proof here for completeness. First of all, by the reverse of Holder's inequality
that applies for $p<1$,
\begin{equation*}
\int f(x)h(x)dx\geq\|f\|_p\|h\|_{p'},
\end{equation*}
where $\frac{1}{p}+\frac{1}{p'}=1$ (here $p'<0$) and $f,h$ are non-negative. When $h(x)=\alpha f^{p-1}(x)$, with $\alpha>0$, there is equality. Hence we have:
\ben
\|f\|_p=\inf_{\|h\|_{p'}=1, h\geq0}\int f(x)h(x)dx.
\een
Applying this to $f_1\star f_2\star\cdot\cdot f_k(x)$, we get:
\ben
\|f_1\star f_2\star\cdot\cdot f_k(x)\|_p=\inf_{\|h\|_{p'}=1,h\geq0}\int f_1\star f_2\star\cdot\cdot f_k(x)h(x)dx.
\een
Now define the spherically increasing symmetric rearrangement $^*h$ of $h$ by:
\ben
^*h=\frac{1}{(\frac{1}{h})^*}.
\een
Then $\|^*h\|_{p'}=\|h\|_{p'}$. For $A>0$, define:
\ben
h_A(x)=\min(A,h(x));k^A(x)=A-h_A(x).
\een
Then, as $A\rightarrow+\infty$,
\ben
h_A(x)\uparrow h(x),
\een
\be\label{eq:explain}
A-(k_A(x))^*\uparrow {^*h}.
\ee
By monotone convergence, we obtain
\ben\begin{split}
\int f_1\star f_2\star\cdot\cdot f_k(x)h(x)dx
&=\lim_{A\uparrow\infty}\int f_1\star f_2\star\cdot\cdot f_k(x)h_A(x)dx \\
&=\lim_{A\uparrow\infty}A-\int f_1\star f_2\star\cdot\cdot f_k(x)k_A(x)dx.
\end{split}\een
Now, similar to the proof of Case 1, we can apply Theorem~\ref{thm:BLL} to obtain
\be\label{eq:lem-f-inter1}\begin{split}
\int f_1\star f_2\star\cdot\cdot f_k(x)h(x)dx.
&\geq\lim_{A\uparrow\infty}A-\int f_1^*\star f_2^*\star\cdot\cdot f_k^*(x)(k_A)^*(x)dx \\
&=\lim_{A\uparrow\infty}\int f_1^*\star f_2^*\star\cdot\cdot f_k^*(x)(A-(k_A)^*)(x)dx \\
&=\int f_1^*\star f_2^*\star\cdot\cdot f_k^*(x) \, [{^*h(x)}]\, dx \\
&\geq\|f_1^*\star f_2^*\star\cdot\cdot f_k^*\|_p,
\end{split}\ee
where the last step follows from duality. Again by duality, taking an infimum over
nonnegative functions $h$ with $\|h\|_{p'}=1$ in inequality \eqref{eq:lem-f-inter1}
gives us
\ben\begin{split}
\|f_1\star f_2\star\cdot\cdot f_k(x)\|_p
&\geq\|f_1^*\star f_2^*\star\cdot\cdot f_k^*(x)\|_p.
\end{split}\een

\noindent {\bf Case 4: $p=0$}.\\
This follows from Case 3 and Lemma~\ref{lem:key}.

\noindent {\bf Case 5: $p=1$}.\\
If $h(f_1\star f_2\star\cdot\cdot f_k)$ is well defined, then by Lemma~\ref{lem:key}(iv), we have:
\ben
h(f_1\star f_2\star\cdot\cdot f_k)\geq h_p(f_1\star f_2\star\cdot\cdot f_k), \,\forall1<p<+\infty.
\een
By Case 1, we have
\begin{equation*}
h(f_1\star f_2\star\cdot\cdot f_k)\geq h_p(f_1^* \star f_2^*\star\cdot\cdot f_k^*), \,\forall1<p<+\infty.
\end{equation*}
Taking the limit as $p\downarrow 1$, we obtain
\ben
h(f_1\star f_2\star\cdot\cdot f_k)\geq h^{+}_1(f_1^* \star f_2^*\star\cdot\cdot f_k^*).
\een
The other inequality is proved similarly.

Finally, if $\|f_1\|_{1+\epsilon}<+\infty$, by Lemma~\ref{lem:preser-ent} and Young's inequality, we have:
\ben\begin{split}
\|f_1^* \star f_2^*\star\cdot\cdot f_k^*\|_{1+\epsilon} &< +\infty,\\
\|f_1\star f_2\star\cdot\cdot f_k\|_{1+\epsilon} &< +\infty.
\end{split}\een
Now we can apply Case 5, Corollary~\ref{cor:easy} and Lemma~\ref{lem:key}(ii) to conclude.
\end{proof}
\vspace{.1in}

\begin{rmk}
In \cite{BL76b}, the inequality \eqref{eq:explain} was claimed without proof. We sketch a proof here. Let
\ben
m(x)=\frac{1}{\lim_{A\uparrow+\infty}(A-k_A^*(x))}.
\een
To show $m(x)=(\frac{1}{h})^*$, we show that $\{x:m(x)>t\}=\{x:(\frac{1}{h})^*>t\}$. But
\ben
\{x:m(x)>t\}=\bigcup_{n=1}^{\infty}\bigcap_{M=1}^{\infty}\bigg\{x:(M-k_M^*(x))<\frac{1}{t+\frac{1}{n}}\bigg\}.
\een
Using Lemma~\ref{lem:basic}, we obtain:
\ben\begin{split}
\{x:m(x)>t\} &= \bigcup_{n=1}^{\infty}\bigcap_{M=1}^{\infty} \bigg\{x:h_M(x)<\frac{1}{t+\frac{1}{n}}\bigg\}^* \\
&=\bigcup_{n=1}^{\infty}\bigcap_{M=1}^{\infty} \bigg\{x:\bigg(\frac{1}{h_M}\bigg)^*>t+\frac{1}{n}\bigg\} \\
&=\bigg\{x: \bigg(\frac{1}{h}\bigg)^*>t\bigg\},
\end{split}\een
where the last step follows since:
\ben
\bigg(\frac{1}{h_M}\bigg)^*=\frac{1}{M}+\int_{\frac{1}{M}}^{\infty} \bigg\{\frac{1}{h(x)}>t\bigg\}^* dt
\,\downarrow\, \bigg(\frac{1}{h}\bigg)^*,
\een
as $M$ goes to $\infty$, where the equality follows directly from definition.
\end{rmk}
\vspace{.1in}


We now prove a variation of Theorem~\ref{thm:main1} when $p=1$.
Although we will prove the best possible version of the $p=1$ case (as stated in the
Main Theorem in Section~\ref{sec:intro}) later, the method of proof of the following seems interesting.

\begin{thm}\label{thm:mainver2}
Let $f_i, 1\leq i\leq k$ be probability densities on $\mathbb{R}^n$ and
$f^*_i, 1\leq i\leq k$, be their respective spherically symmetric decreasing  rearrangements.
If
\ben
\int f_1\star f_2\star\cdot\cdot\star f_k(x)\|x\|^2dx<+\infty ,
\een
then
\ben
h(f_1\star f_2\star\cdot\cdot\cdot\star f_k)\geq h(f^*_1\star f^*_2\star\cdot\cdot\cdot\star f^*_k).
\een
\end{thm}

\begin{proof}
By Lemma~\ref{lem:finvar},
\ben
\int f_1^*\star f_2^*\star\cdot\cdot\star f_k^*(x)\|x\|^2dx<+\infty.
\een
Hence both $h(f_1\star f_2\star\cdot\cdot\cdot\star f_k)$ and $h(f^*_1\star f^*_2\star\cdot\cdot\cdot\star f^*_k)$ are well defined. Clearly, the Main Theorem implies that
\be\label{eq:ver1key3}
h(\sqrt{t}Z+X_1+\cdot\cdot+X_k)\geq h(\sqrt{t}Z+X_1^*+\cdot\cdot+X_k^*),
\ee
where $t>0$, $Z$ is a $n$ dimensional standard normal, $X_i$ is distributed according to $f_i$, $X_i^*$ is distributed according to $f_i^*$ and all random vectors are independent. The rest of the argument follows by taking the limit as $t$ goes to $0$. To simplify the notation, let
\ben
X=X_1+\cdot\cdot+X_k,X_*=X_1^*+\cdot\cdot+X_k^*.
\een
The joint lower-semicontinuity of relative entropy functional \cite{DZ98:book} gives
\ben
D(X\|G)\leq\liminf_{t\downarrow0}D(\sqrt{t}Z+X\|\sqrt{t}Z+G),
\een
where $G$ is a Gaussian random vector matching the mean and covariance of $X$, independent of $Z$.
But by a standard equality in information theory \cite{CT06:book}, we get
\ben
D(X\|G)=h(G)-h(X),
\een
\ben
D(\sqrt{t}Z+X\|\sqrt{t}Z+G)=h(\sqrt{t}Z+G)-h(\sqrt{t}Z+X).
\een
It's also easy to check directly that
\ben
\lim_{t\downarrow0}h(\sqrt{t}Z+G)=h(G).
\een
Hence we obtain
\be\label{eq:ver1key1}
\limsup_{t\downarrow0}h(\sqrt{t}Z+X)\leq h(X).
\ee
On the other hand, the density  $f_{t*}$ of $\sqrt{t}Z+X_*$ can be expressed as
\ben
\mathbb{E}f_*(x-\sqrt{t}Z),
\een
with $f_*$ being the density of $X_*$. If we apply Jensen's inequality to the concave function $-u\log(u)$, we obtain
\be\label{eq:ver1local}
-f_{t*}(x)\log(f_{t*}(x))
\geq-\mathbb{E}f_*(x-\sqrt{t}Z)\log(f_*(x-\sqrt{t}Z)).
\ee
It is easy to check the right hand side of the above is well defined due to the boundedness of the normal density
and the finiteness of the second moment of $X_*$. If $h(X_*)=-\infty$, we trivially have
\ben
h(X_*+\sqrt{t}Z)\geq h(X_*).
\een
Hence we can assume $h(X_*)$ is finite. In this case, we can integrate both sides of \eqref{eq:ver1local}
with respect to $x$ and use Fubini's Theorem to conclude that
\be\label{eq:ver1key2}
h(X_*+\sqrt{t}Z)\geq h(X_*).
\ee
Combining the inequalities \eqref{eq:ver1key1}, \eqref{eq:ver1key2} and \eqref{eq:ver1key3}, we can conclude the proof.
\end{proof}
\vspace{.1in}

\begin{rmk}\label{rmk:entconti}
We remark that the second part of the proof can be applied to $X$, instead of to $X_*$, as well. Hence,
under the only assumption of the finiteness of the second moment, we obtain:
\ben
\lim_{t\downarrow 0}h(X+\sqrt{t}Z)=h(X).
\een
This is known implicitly in \cite{CS91}, with a slightly different proof, but never seems
to have been explicitly noted. Of course, the continuity of $h(X+\sqrt{t}Z)$ in $t$ is trivial when $t>0$
because then one already has as much smoothness as desired to start with.
\end{rmk}


\section{Second Proof of Main Theorem via Majorization}
\label{sec:major}

In this section, we give a new and unified proof of the Main Theorem for all values of $p$, using ideas from majorization theory.
In particular, we will show a best possible version of the $p=1$ case of the Main Theorem. A generalization of the Main Theorem
is also be obtained. We first define majorization as a partial order on the set of densities.

\begin{defn}
For probability densities $f$ and $g$ on $\mathbb{R}^n$, we say that $f$ is majorized by $g$
if
\ben
\int_{\{x:\|x\|<r\}}f^*(x)dx\leq\int_{\{x:\|x\|<r\}}g^*(x)dx
\een
for all $r>0$. In this case, we write $f\prec g$.
\end{defn}

We also need the following lemma, which includes Lemma~\ref{lem:preser-ent} as a special case.

\begin{lem}\label{lem:allpreser}
Let $f$ be a probability density and $\phi(x)$ be a convex function defined on the non-negative real line
such that $\phi(0)=0$ and it is continuous at $0$. Then
\ben
\int\phi(f(x))dx=\int\phi(f^*(x))dx,
\een
provided that one of these integrals is well defined (which guarantees that the other is).
\end{lem}

\begin{proof}
Note that a convex function satisfying the assumed conditions is always absolutely continuous.
If $\phi(x)\geq0,\forall x\geq0$, then $\phi$ must be increasing and we have $\phi'\geq 0$;
if $\phi(x)\leq0,\forall x\geq0$, then $\phi$ must be decreasing and we have $\phi'\leq0$.
In both cases, straightforward applications of Tonelli's theorem and the layer cake representation will do.

For the remaining case, we can assume that there is a unique $\alpha>0$, such that $\phi(\alpha)=0$,
$\phi(x)>0,x>\alpha$ and $\phi(x)\leq0,x<\alpha$. Noting that $\phi^+$ is also a continuous convex function null at zero,
by what has been proved, we obtain
\be\label{eq:plus}
\int\phi^+(f(x))dx=\int\phi^+(f^*(x))dx.
\ee
On the other hand, we have
\ben
\phi^-(u)+\phi'_{+}(\alpha)u\{u\leq\alpha\} = \{u\leq\alpha\}\int_0^{u}(\phi'_{+}(\alpha)-\phi'(t))_+dt.
\een
Tonelli's theorem gives us
\ben
\int\bigg[\phi^-(f(x))+\phi'_{+}(\alpha)f(x)\{f(x)\leq\alpha\}\bigg]dx
= \int\bigg[\phi^-(f^*(x))+\phi'_{+}(\alpha)f^*(x)\{f^*(x)\leq\alpha\}\bigg]dx.
\een
But it is easy to see that
\ben
-\infty<\int\phi'_{+}(\alpha)f(x)\{f(x)\leq\alpha\}dx
= \int\phi'_{+}(\alpha)f^*(x)\{f^*(x)\leq\alpha\}dx<+\infty.
\een
Hence
\be\label{eq:minus}
\int\phi^-(f(x))dx=\int\phi^-(f^*(x))dx ,
\ee
and we can conclude, provided the quantities in \eqref{eq:plus} and \eqref{eq:minus} are not both $+\infty$ or both $-\infty$.
\end{proof}
\vspace{.1in}

The following is well known in majorization theory. In the discrete case,
it is first proved in Hardy, Littlewood and Polya \cite{HLP29}. Various extensions to the continuous
setting are discussed by Chong \cite{Cho74}. However, we are not able to find a
direct reference that covers all cases of our interest. So we give the proof here.

\begin{prop}\label{prop:gen-id}
Let $\phi(x)$ be a convex function defined on the non-negative real line such that $\phi(0)=0$ and it is continuous at $0$.
If $f$ and $g$ are probability densities, with $f\prec g$, then
\ben
\int\phi(f(x))dx\leq\int\phi(g(x))dx,
\een
provided that both sides are well defined.
\end{prop}

\begin{proof}
We first show that for each $t>0$,
\be\label{eq:proofkey}
\int(f(x)-t)_+dx\leq\int(g(x)-t)_+dx.
\ee
By Lemma~\ref{lem:allpreser}, we only need to show:
\ben
\int(f^*(x)-t)_+dx\leq\int(g^*(x)-t)_+dx.
\een
By Markov's inequality, we know that the set $\{x:f^*(x)>t\}$ is an open ball with finite radius, say $r$. Then we have
\ben\begin{split}
\int(f^*(x)-t)_+dx &= \int_{\{x:\|x\|<r\}}(f^*(x)-t)dx \\
&\leq\int_{\{x:\|x\|<r\}}(g^*(x)-t)dx \\
&\leq\int(g^*(x)-t)_+dx.
\end{split}\een
Next, we assume that, additionally, $\phi'_+(0)$ is finite. Define the second derivative measure $\mu$ of $\phi$ by setting
\ben
\mu((a,b])=\phi'_+(b)-\phi'_+(a),
\een
\ben
\mu(\{0\})=\phi'_+(0).
\een
Restricted to $(0,+\infty)$, $\mu$ is a non-negative measure. Using Tonelli's theorem, we see:
\ben
\phi(u)=\int_0^{u}\phi'(t)dt
\een
\ben
=\int_0^{u}\{0<s\leq t\}\mu(ds)dt+\phi'_+(0)u
\een
\ben
=\int_{(0,+\infty)}(u-s)_+\mu(ds)+\phi'_+(0)u.
\een
Also, since $\phi^-(x)\leq|\phi'(0)|x$, $\int\phi(f(x))dx$ is well defined for all density $f$. Hence by integrating both sides of \eqref{eq:proofkey} with respect to $\mu$ on $(0,+\infty)$ and using Tonelli's theorem, we obtain:
\ben
\int\phi(f(x))dx\leq\int\phi(g(x))dx.
\een
Finally, if $\phi'_+(0)$ is not finite, then it must be $-\infty$ and we can find a $\alpha>0$ such that $\phi(\alpha)<0$. Define $\phi_n(x)=\phi(\frac{\alpha}{2^n})\frac{2^n}{\alpha}x$, if $x\leq\frac{\alpha}{2^n}$ and $\phi_n(x)=\phi(x)$ otherwise. Then
\ben
\phi_n^+(x)=\phi^+(x),
\een
\ben
\phi_n^-(x)\uparrow\phi^-(x).
\een
By what has been proved,
\ben
\int\phi_n(f(x))dx\leq\int\phi_n(g(x))dx,
\een
which is equivalent to,
\ben
\int\phi^+(f(x))dx-\int\phi_n^-(f(x))dx
\een
\ben
\leq\int\phi^+(g(x))dx-\int\phi_n^-(g(x))dx.
\een
If we assume that both $\int\phi(f(x))dx$ and $\int\phi(g(x))dx$ are well defined, then we can use monotone convergence to conclude.
\end{proof}
\vspace{.1in}

We now apply this to obtain a proof of the Main Theorem (indeed, a generalization of it) under minimal assumptions.

\begin{thm}\label{thm:most-gen}
Suppose $f_i, 1\leq i \leq k$ are probability densities. Let $\phi(x)$ be a convex function defined on the non-negative real line
such that $\phi(0)=0$ and $\phi$ is continuous at $0$. Then
\ben
\int\phi(f_1\star f_2\star\cdot\cdot\star f_k(x))dx\leq\int\phi(f_1^*\star f_2^*\star\cdot\cdot\star f_k^*(x))dx,
\een
provided that both sides are well defined.
\end{thm}

\begin{proof}
We only need to show $f=f_1\star f_2\star\cdot\cdot\star f_k\prec g=f_1^*\star f_2^*\star\cdot\cdot\star f_k^*$.
Again by Theorem~\ref{thm:BLL}, we have:
\be\label{eq:bll}
\int r(x)f(x)dx\leq\int r^*(x)g(x)dx.
\ee
Now we recall the following representation \cite{Bur09:tut}:
\ben
\int_{B(0,r)}f^*(x)dx=\sup_{|C|=|B(0,r)|}\int_C f(x)dx,
\een
for any density $f$, where $B(0,r)$ is the open ball with radius $r$. By \eqref{eq:bll}, we have that
\ben
\int_{B(0,r)}f^*(x)dx\leq\sup_{|C|=|B(0,r)|}\int_{C^*}g(x)dx.
\een
But
\ben
\sup_{|C|=|B(0,r)|}\int_{C^*}g(x)dx\leq\sup_{|A|=|B(0,r)|}\int_{A}g(x)dx,
\een
since $|C^*|=|C|$. Using the representation again, we obtain:
\ben
\int_{B(0,r)}f^*(x)dx\leq\int_{B(0,r)}g^*(x)dx.
\een
\end{proof}
\vspace{.1in}

\begin{rmk}
By taking $\phi(x)=x^p$ for $p>1$, $\phi(x)=-x^p$ for $0<p<1$ and $\phi(x)=x\log(x)$ for $p=1$, we recover the Main Theorem.
\end{rmk}

\section{Implications}
\label{sec:implicn}

In this section, we point out several implications of the Main Theorem, some of which have already been mentioned in the Introduction.
The first implication is the Brunn-Minkowski inequality. One can recover it from the Main Theorem in full generality.

\begin{cor}
Let $A$ and $B$ be two nonempty Borel sets. Then:
\ben
|A+B|^\nth \geq |A^* + B^*|^\nth = |A|^\nth + |B|^\nth .
\een
\end{cor}

\begin{proof}
We first assume both of them have non-zero and finite volume.
We take $f_1=\frac{\mathbb{I}_{A}}{|A|}$ and $f_2=\frac{\mathbb{I}_{B}}{|B|}$. Then
\begin{equation*}
f^*_1=\frac{\mathbb{I}_{A^*}}{|A|},f^*_2=\frac{\mathbb{I}_{B^*}}{|B|}.
\end{equation*}
By the $p=0$ case of the Main Theorem and Lemma~\ref{lem:key}(i), we have:
\begin{equation*}
|supp(f_1\star f_2)|\geq|supp(f^*_1\star f^*_2)|.
\end{equation*}
As something that can be checked easily,
\begin{equation*}
supp(f_1\star f_2)\subseteq A+B,supp(f^*_1\star f^*_2)=A^*+B^*.
\end{equation*}
Hence we can conclude. If one of them has zero volume, say $B$, we take a point $x_0\in B$. Then:
\ben
|A+B|^\nth=|A+B-x_0|^\nth\geq|A|^\nth,
\een
where the inequality follows since $A+B-x_0\supseteq A$. If one of them has infinite volume, say $B$, we take a point $x_0\in A$. Then:
\ben
|A+B|^\nth=|A+B-x_0|^\nth\geq|B|^\nth.
\een
\end{proof}
\vspace{.1in}

We next derive Theorem~\ref{thm:BLL} from the Main Theorem, thus showing that they are mathematically
equivalent to each other. However, to recover Theorem~\ref{thm:BLL} for $k$ functions, we need the Main Theorem
for $k$ densities, while, if we look back at the proof of the Main Theorem for $k$ densities,
we need Theorem~\ref{thm:BLL} for $k+1$ functions. To summarize the following,
Theorem~\ref{thm:BLL} for $k$ functions can be seen as the $p=+\infty$ case of the Main Theorem for $k$ densities.

\begin{cor}
Let $A_i, 1\leq i\leq k$ be measurable subsets of $\RL^n$ with finite volume. Then
\ben
\star_{i\in[k]}\mathbb{I}_{A_i}(0) \leq\star_{i\in[k]}\mathbb{I}_{A^*_i}(0).
\een
\end{cor}

\begin{proof}
Clearly, we can assume all the sets have non-zero volume. We first prove the corollary under the assumption that one of the sets, say $A_1$, is open.  But the indicator of an open set is an lower semicontinuous function. Hence by Fatou's lemma,
\ben\begin{split}
\liminf_{x_m\rightarrow x}\int\mathbb{I}_{A_1}(x_m-y)\star_{2\leq i\leq k}\mathbb{I}_{A_i}(y)dy
&\geq\int\liminf_{x_m\rightarrow x}\mathbb{I}_{A_1}(x_m-y)\star_{2\leq i\leq k}\mathbb{I}_{A_i}(y)dy \\
&\geq\int\mathbb{I}_{A_1}(x-y)\star_{2\leq i\leq k}\mathbb{I}_{A_i}(y)dy,
\end{split}\een
which shows that $\star_{i\in[k]}\mathbb{I}_{A_i}(x)$ is lower-semicontinuous. Similarly, $\star_{i\in[k]}\mathbb{I}_{A^*_i}(x)$ is also lower-semicontinuous. Now the claimed result follows from the $p=\infty$ case of the Main Theorem, Lemma~\ref{lem:key}(iii) and the classical fact that
\ben
\star_{i\in[k]}\mathbb{I}_{A^*_i}(0)=\sup_x\star_{i\in[k]}\mathbb{I}_{A^*_i}(x).
\een
In the general case, by the regularity property of the Lebesgue measure, we can find a sequence of open set $B_m\supseteq A_1$, such that
\ben
\lim_{m\rightarrow\infty}|B_m|=|A_1|,
\een
from which it follows that
\be\label{eq:localrs1}
\lim_{m\rightarrow\infty}\|\mathbb{I}_{B_m}-\mathbb{I}_{A_1}\|_1=0.
\ee
By Lemma~\ref{lem:L1}, we also have:
\be\label{eq:localrs2}
\lim_{m\rightarrow\infty}\|\mathbb{I}_{B_m^*}-\mathbb{I}_{A_1^*}\|_1=0.
\ee
The inequalities \eqref{eq:localrs1} and \eqref{eq:localrs2} imply easily that for each $x$,
\ben
\int\mathbb{I}_{B_m}(x-y)\star_{2\leq i\leq k}\mathbb{I}_{A_i}(y)dy\rightarrow\star_{i\in[k]}\mathbb{I}_{A_i}(x),
\een
\ben
\int\mathbb{I}_{B_m^*}(x-y)\star_{2\leq i\leq k}\mathbb{I}_{A_i^*}(y)dy\rightarrow\star_{i\in[k]}\mathbb{I}_{A_i^*}(x).
\een
By what has been proved,
\ben\begin{split}
\int\mathbb{I}_{B_m}(-y)\star_{2\leq i\leq k}\mathbb{I}_{A_i}(y)dy
&\leq\int\mathbb{I}_{B_m^*}(-y)\star_{2\leq i\leq k}\mathbb{I}_{A_i^*}(y)dy.
\end{split}\een
Hence the desired result follows by taking the limit as $m$ goes to $\infty$.
\end{proof}
\vspace{.1in}

\begin{rmk}
By taking $A_1$ to be $-A_1$, we obtain
\ben\begin{split}
\int\mathbb{I}_{A_1}(y)\star_{2\leq i\leq k}\mathbb{I}_{A_i}(y)dy
&\leq\int\mathbb{I}_{A_1^*}(y)\star_{2\leq i\leq k}\mathbb{I}_{A_i^*}(y)dy.
\end{split}\een
For any $k$ densities, $f_i, 1\leq i\leq k$, we have the following layer cake representation:
\ben\begin{split}
\int f_1(y)\star_{2\leq i\leq k}f_i(y)dy
&=\int_0^{\infty}dt_1\cdot\cdot\int_0^{\infty}dt_k\int\mathbb{I}_{A_{t_1}^1}(y)\star_{2\leq i\leq k}\mathbb{I}_{A_{t_i}^i}(y)dy,
\end{split}\een
where $A_{t_i}^i=\{x:f_i(x)>t_i\}$. This, combined with Lemma~\ref{lem:basic}, will give us:
\ben\begin{split}
\int f_1(y)\star_{2\leq i\leq k}f_i(y)dy
&\leq\int f_1^*(y)\star_{2\leq i\leq k}f_i^*(y)dy.
\end{split}\een
This is precisely Theorem~\ref{thm:BLL}.
\end{rmk}

Another corollary of the Main Theorem is a Fisher information inequality. {\bf From now on,
until the end of this section, we assume $n=1$.} We first define Fisher information of a probability measure.

\begin{defn}
If a density $f$ on the real line is locally absolutely continuous, with the derivative $f'$ (defined almost everywhere),
then its Fisher information $I(f)$ is defined by
\begin{equation*}
I(f)=\int_{\{x:f(x)>0\}}\frac{f'^2(x)}{f(x)}dx.
\end{equation*}
For other densities and for probability measures without densities we define $I$ to be $+\infty$.
\end{defn}

We will sometimes abuse notation by writing $I(X)$ to mean $I(f)$, if $X$ is distributed according to $f$.
One can show that if $I$ is finite, then the derivative of the density $f$ is absolutely integrable \cite{BCG14:fisher}.
We also need some important properties of Fisher information.
One is that Gaussian convolution decreases Fisher information. This is a slight extension of the argument in \cite{Bar84:tr}.
We give this as a lemma and give a complete proof.

\begin{lem}\label{lem:fishcov}
\begin{equation*}
I(X+G)\leq I(X),
\end{equation*}
where  $X$ is any random variable with a density $f$ and $G$ is a non-degenerate Gaussian with density $g$, independent of $X$.
\end{lem}

\begin{proof}
Clearly, we can assume $I(X)<+\infty$. Let $S=X+G$, with density $h(x)$. It is easy to see that $h(x)$ is strictly positive and differentiable, with
\ben
h'(x)=\mathbb{E}g'(x-X),
\een
which can be justified by dominated convergence. By the finiteness of Fisher information, $f'$ is absolutely integrable. Hence
\ben
h'(x)
=\int\int g'(x-z)\{t\leq z\}f'(t)dt dz
=\int g(x-t)f'(t)dt,
\een
which can justified by Fubini's Theorem. We now show that
\be\label{eq:locfcov}
\frac{h'(S)}{h(S)}=\mathbb{E}\bigg[\frac{f'(X)\{f(X)>0\}}{f(X)}\bigg|S\bigg].
\ee
Once \eqref{eq:locfcov} is shown, we can apply conditional version of Jensen's inequality to obtain:
\ben
\bigg(\frac{h'(S)}{h(S)}\bigg)^2\leq\mathbb{E}\bigg[\bigg(\frac{f'(X)\{f(X)>0\}}{f(X)}\bigg)^2\bigg|S\bigg],
\een
which, after taking expectations, will give us the desired result. To show \eqref{eq:locfcov}, it suffices to show:
\be\label{eq:locfcov2}
\mathbb{E}\bigg[m(S)\frac{h'(S)}{h(S)}\bigg]=\mathbb{E}\bigg[\frac{f'(X)\{f(X)>0\}}{f(X)}m(S)\bigg],
\ee
where $m(x)$ is any bounded measurable function. The left side of inequality \eqref{eq:locfcov2} is
\ben\begin{split}
\int m(x)h'(x)dx
&=\int\int m(x)f'(t)g(x-t) dx dt \\
&=\int\int m(z+t)\frac{f'(t)\{f(t)>0\}}{f(t)}f(t)g(z) dz dt \\
&\quad +\int\int m(z+t)f'(t)\{f(t)=0\}g(z) dz dt,
\end{split}\een
where we have used Fubini's Theorem implicitly. The first term on the right side of the above display is precisely the right side
of the inequality \eqref{eq:locfcov2}, while the fact that the second term is zero follows from \eqref{eq:psrmk}.
\end{proof}
\vspace{.1in}

The other one is the lower semicontinuity of Fisher information functional with respect to weak convergence
topology \cite{BCG14:fisher}. Lemma~\ref{lem:fishcov} and the lower semicontinuity allow us to extend
\cite[Lemma 6.1]{Bar84:tr} to $t=0$, which is of some independent interest.

\begin{lem}\label{lem:debext}
Assume that $X$ has finite second moment, and a density with $h(X)>-\infty$. Then
\begin{equation*}
\frac{d h(X+\sqrt{t}Z)}{dt}\bigg|_{t=0}=\frac{1}{2}I(X),
\end{equation*}
where $I(X)$ might be infinity and $Z$ is a standard normal, independent of $X$.
\end{lem}

\begin{proof}
For $t>0$, Lemma 6.1 in  $\cite{Bar84:tr}$ implies that
\begin{equation*}
\frac{d h(X+\sqrt{t}Z)}{dt}=\frac{1}{2}I(X+\sqrt{t}Z).
\end{equation*}
Remark~\ref{rmk:entconti} shows that $h(X+\sqrt{t}Z)$ is continuous at $t=0$. Hence we can apply the mean value Theorem to obtain:
\ben
\lim_{k\rightarrow\infty}\frac{h(X+\sqrt{t_k}Z)-h(X)}{t_k}=\frac{1}{2}\lim_{k\rightarrow\infty}I(X+\sqrt{s_k}Z),
\een
where $t_k\downarrow0,s_k\downarrow0$.
Now by Lemma~\ref{lem:fishcov}, $I(X+\sqrt{s_k}Z)\leq I(X)$. This, combined with lower-semicontinuity of Fisher information, allows us to conclude.
\end{proof}
\vspace{.1in}

The final one is the continuity of the Fisher information functional after Gaussian convolution (see \cite{Joh04:book}).

\begin{lem}\label{lem:fishcon}
Let $G$ be a non-degenerate Gaussian, independent of $X$. Then
$I(X+G)$ is a continuous functional of the density of $X$, with respect to the topology of weak convergence.
\end{lem}

\begin{cor}\label{cor:fishdec}
For any density $f$, let $X$ be a random vector distributed according to $f$ and $X^*$ distributed according to $f^*$.
Then
\begin{equation*}
I(X)\geq I(X^*).
\end{equation*}
\end{cor}

\begin{proof}
We first assume that $X$ has finite second moment and that $h(X)>-\infty$. The $p=1$ case of the Main Theorem implies that:
\begin{equation*}
h(X+\sqrt{t}Z)\geq h(X^*+\sqrt{t}Z).
\end{equation*}
where $Z$ is a standard normal and all random variables are independent.
By Lemma~\ref{lem:preser-ent} and our assumption,
\begin{equation}
-\infty<h(X)=h(X^*)<+\infty.
\end{equation}.
This further implies:
\begin{equation*}
\frac{h(X+\sqrt{t}Z)-h(X)}{t}\geq\frac{h(X^*+\sqrt{t}Z)-h(X^*)}{t}.
\end{equation*}
Now we can apply Lemma~\ref{lem:debext} to conclude (note that by Lemma~\ref{lem:finvar}, $X^*$ also has finite second moment).

In the general case, we use an approximation argument. Specifically, note that continuous functions with compact support is dense in the space $L_1(\mathbb{R})$. Hence we can find $f_n\in C_c(\mathbb{R})$, such that:
\begin{equation*}
\|f_n-f\|_1\rightarrow0.
\end{equation*}
Then we have
\begin{equation*}
\|f^+_n-f\|_1\rightarrow0.
\end{equation*}
and
\begin{equation*}
\bigg\|\frac{f^+_n}{\|f^+_n\|_1}-f\bigg\|_1\rightarrow0.
\end{equation*}
Hence we conclude we can find a sequence of densities $g_n\in C_c(\mathbb{R})$, such that:
\begin{equation*}
\|g_n-f\|_1\rightarrow0.
\end{equation*}
We now show that
\be\label{eq:localps1}
I(X_n+\sqrt{t}Z)\geq I((X_n+\sqrt{t}Z)^*),
\ee
where $X_n$ has density $g_n$ and all random variables are independent. Clearly, $X_n+\sqrt{t}Z$ has finite second moment. Since the density of $X_n+\sqrt{t}Z$ is bounded, $h(X_n+\sqrt{t}Z)>-\infty$. Hence \eqref{eq:localps1} follows from what has been proved.

Finally, taking the limit as $n$ goes to infinity and applying Lemma~\ref{lem:fishcon} to the left of \eqref{eq:localps1} and lower-semicontinuity to the right of \eqref{eq:localps1}, we see that (note that the density of $(X_n+\sqrt{t}Z)^*$ converges to that of $(X+\sqrt{t}Z)^*$ in total variation distance due to Lemma~\ref{lem:L1}):
\be\label{eq:localps2}
I(X+\sqrt{t}Z)\geq I((X+\sqrt{t}Z)^*).
\ee
We then apply Lemma~\ref{lem:fishcov} to the left of \eqref{eq:localps2} and lower semicontinuity to the right of \eqref{eq:localps2} to obtain (note that the density of $(X+\sqrt{t}Z)^*$ converges to that of $X^*$ in total variation distance due to Lemma~\ref{lem:L1}):
\begin{equation*}
I(X)\geq I(X^*),
\end{equation*}
by taking the limit as $t$ goes to $0$.
\end{proof}
\vspace{.1in}

\begin{rmk}
Two standard facts are used implicitly about Gaussian convolution in the above:
\begin{equation*}
\|g_t\star g_n-g_t\star f\|_1\leq\|g_n-f\|_1;
\end{equation*}
\begin{equation*}
\|g_t\star f-f\|_1\rightarrow0,
\end{equation*}
as $t$ goes to $0$.
\end{rmk}

\begin{rmk}
The above inequality is completely equivalent to the Polya-Szego inequality for $p=2$ \cite{Bur09:tut}.
Suppose that a non-negative function $g$, locally absolutely continuous, satisfies $\|g\|_2<\infty$ and $\|g'\|_2<\infty$
($g\in\mathbf{H}_2(\mathbb{R})$). We assume for now that $\|g\|_2=1$. Then $f=g^2$ is a probability density and
\ben
I(f)=4\int_{\{x:g(x)>0\}}{g'^2(x)}dx<+\infty.
\een
By the above corollary and Lemma~\ref{lem:increasing}, we know that $f^*$ is also absolutely continuous and that
\ben\begin{split}
f^*&=(g^*)^2, \\
I(f^*)&=\int_{\{x:f^*(x)>0\}}\frac{({f^*}'(x))^2}{f^*(x)}dx\leq I(f)<+\infty.
\end{split}\een
We now show that $g^*=\sqrt{f^*}$ is locally absolutely continuous. Fix $\epsilon>0$, then
\ben
\sqrt{f^*(b)+\epsilon}-\sqrt{f^*(a)+\epsilon}=\frac{1}{2}\int_{a}^{b}\frac{{f^*}'}{\sqrt{f^*+\epsilon}}dx.
\een
But for any non-negative measurable function $f$, if $f$ is differentiable at $x_0$ such that $f(x_0)=0$, then 
we must have
\be\label{eq:psrmk}
f'(x_0)=0.
\ee
(This fact appears as \cite[Proposition 2.1]{BCG14:fisher}, with a complicated proof, but it is actually a simple
consequence of the definitions. Indeed, non-negativity gives us that the left derivative $\partial_{-}f(x_0)\leq 0$ 
and the right derivative $\partial_{+}f(x_0)\geq 0$, while differentiability tells us that $\partial_{-}f(x_0)=\partial_{+}f(x_0)$,
so that both are 0, and hence so is $f'(x_0)$.)
Hence
\ben
\sqrt{f^*(b)+\epsilon}-\sqrt{f^*(a)+\epsilon}=\frac{1}{2}\int_{a}^{b}\{f^*(x)>0\}\frac{{f^*}'(x)}{\sqrt{f^*(x)+\epsilon}}dx.
\een
By dominated convergence (using the finiteness of the Fisher information), we obtain
\ben
\sqrt{f^*(b)}-\sqrt{f^*(a)}=\frac{1}{2}\int_{a}^{b}\{f^*(x)>0\}\frac{{f^*}'(x)}{\sqrt{f^*(x)}}dx,
\een
which implies that $g^*$ is locally absolutely continuous. Hence,
\ben
I(f^*)=4\int_{\{x:g^*(x)>0\}}{{{g^*}'(x)}^2}dx.
\een
Finally, the argument that leads to \eqref{eq:psrmk} shows also that
\ben\begin{split}
I(f^*)&=4\int{{g^*}'^2}dx, \\
I(f) &=4\int{g'^2}dx.
\end{split}\een
Hence we obtain
\be\label{eq:ps}
\int{{g^*}'^2}dx\leq\int{g'^2}dx.
\ee
By Lemma~\ref{lem:preser-ent}, the assumption $\|g\|_2=1$ can be removed easily.
Thus we have shown that if a non-negative function $g\in\mathbf{H}_2(\mathbb{R})$, then \eqref{eq:ps} is true.
This is precisely the Polya-Szego inequality for $p=2$.
\end{rmk}

\begin{rmk}
Here is another perspective on Corollary~\ref{cor:fishdec}. Entirely similar to its proof, we can deduce the
following inequality from Theorem~\ref{thm:epi} (under the same assumptions of Theorem~\ref{thm:epi}):
\be\label{eq:isop-ent}
I(f)\geq I(g)=\frac{1}{N(f)},
\ee
where $g$ is a centered normal density such that $h(g)=h(f)$ and $N(f)=\frac{1}{2\pi e}e^{2h(f)}$, the entropy power of $f$.
This inequality, first proved by Stam \cite{Sta59} (by combining the entropy power inequality and de Bruijn's identity)
is sometimes called the ``isoperimetric inequality for entropy''. Hence, just as our Main Theorem strengthens the entropy power inequality, 
Corollary~\ref{cor:fishdec} can be seen as a strengthening of the  
isoperimetric inequality for entropy since it inserts $I(f^*)$ in between $I(f)$ and $I(g)$ in the inequality \eqref{eq:isop-ent}.
\end{rmk}

\begin{rmk}
It is a folklore fact that the isoperimetric inequality for entropy is related to the Gaussian logarithmic Sobolev inequality 
(usually attributed to Gross \cite{Gro75}, who developed its
remarkable applications to hypercontractivity and mathematical physics) under a finite variance constraint on $f$. Indeed,  
if $\tilde{g}$ is a normal density with the same mean and variance as $f$, then $h(f)=h(\tilde{g})-D(f\|g)$, which means
that the inequality \eqref{eq:isop-ent} can be rewritten as
\ben
N(\tilde{g}) I(f) \geq e^{2D(f\|\tilde{g})} .
\een
Using the fact that $N(\tilde{g})$ is just the variance of $\tilde{g}$ and hence the variance $\sigma^2_f$ of $f$, we have
\ben
D(f\|\tilde{g}) \leq \half \log \big[ \sigma^2_f I(f) \big] . 
\een
Since $\log x \leq x-1$ for $x>0$, we obtain
\be\label{eq:lsi-ent}
D(f\|\tilde{g}) \leq \half J(f) ,
\ee
where 
\ben
J(f) :=  \sigma^2_f I(f) -1 
\een 
is the standardized Fisher information of $f$, which is location and scale invariant. The inequality \eqref{eq:lsi-ent}
is a form of the Gaussian logarithmic Sobolev inequality. (In a related fashion, Toscani \cite{Tos13}
uses Costa's concavity of entropy power \cite{Cos85b} to prove Nash's inequality with the sharp constant.)
\end{rmk}

\section{An Application: Bounding the entropy of the sum of two uniforms}
\label{sec:appln}

\begin{prop}
Let $X$ and $Y$ be independent uniform distributions on two Borel sets $M_1$ and $M_2$, both with finite, non-zero volume. Then the following is true:
\ben
h(X+Y)\geq h(X^*+Y^*)=\log\bigg(B\bigg(\frac{n+1}{2},\frac{1}{2}\bigg)V_n(1){r_1}^n{r_2}^n\bigg)
+\int_0^{r_1+r_2}\log\bigg(\frac{1}{g(r)}\bigg)\frac{ng(r)r^{n-1}}{{r_1}^n{r_2}^nB(\frac{n+1}{2},\frac{1}{2})}dr,
\een
where $B(\cdot,\cdot)$ is the beta function, $V_n(1)$ is the volume of the $n$ dimensional unit ball and
\begin{equation*}
V_n(1){r_i}^n=|M_i|,
\end{equation*}
for $i=1,2$. The function $g(r)$ is defined in the following way:
if $r>|r_1-r_2|$,
\ben
g(r)={r_1}^nh\bigg(\arcsin\bigg(\frac{r^2-{r_2}^2+{r_1}^2}{2rr_1}\bigg)\bigg)
+{r_2}^nh\bigg(\arcsin\bigg(\frac{r^2-{r_1}^2+{r_2}^2}{2rr_2}\bigg)\bigg),
\een
where $h(\theta)=\int_{\theta}^{\frac{\pi}{2}}{\cos^n(x)}dx$ and if $r<|r_1-r_2|$,
\ben
g(r)=\min({r_1}^n,{r_2}^n)B\bigg(\frac{n+1}{2},\frac{1}{2}\bigg).
\een
\end{prop}

\begin{proof}
By the $p=1$ case of the Main Theorem, we get
\ben
h(X+Y)\geq h(X^*+Y^*),
\een
where $X^*$ and $Y^*$ are independent uniform distributions on the balls centered at the origin,
with radius $r_1$ and $r_2$ respectively. We just need to compute explicitly the density of $X^*+Y^*$. This is given by:
\ben
\frac{1}{|M_1||M_2|}\int\mathbb{I}_{M_1^*}(x-y)\mathbb{I}_{M_2^*}(y)dy.
\een

Note that $\int\mathbb{I}_{M_1^*}(x-y)\mathbb{I}_{M_2^*}(y)dy$ is nothing but the volume of the
intersection of two balls $|B(x,r_1)\cap B(0,r_2)|$.
If $\|x\|<|r_2-r_1|$, this volume is that of the smaller ball;
if $\|x\|\geq r_1+r_2$, this volume is zero;
if $|r_2-r_1|\leq\| x\|\leq r_1+r_2$, this volume is the sum of the volumes of two spherical caps.
Then if $\|x\|=r$, the spherical cap with radius $r_2$ will have height $h_2=\frac{r_1^2-(r-r_2)^2}{2r}$
and the cap with radius $r_1$ will have height $h_1=\frac{r_2^2-(r-r_1)^2}{2r}$.
But the volume of a spherical cap with given radius and height is classical and can be
computed from elementary calculus. The rest follows from simple algebra.
\end{proof}
\vspace{.1in}

\begin{cor}\label{cor:dim1}
If $n=1$ and $|M_2|>|M_1|$, then
\ben\begin{split}
h(X+Y) &\geq h(Y)+\frac{|M_1|}{2|M_2|} \\
&\geq\log\bigg(|M_2|+\frac{1}{2}|M_1|\bigg).
\end{split}\een
\end{cor}

\begin{proof}
The first inequality follows from by doing all the integrations explicitly. The second follows from the simple inequality $e^{x}\geq1+x$.
\end{proof}
\vspace{.1in}

\begin{rmk}
Note that since $X+Y$ is supported in $M_1+M_2$, we have the upper bound
\be\label{eq:localsum1}
|M_1+M_2|\geq e^{h(X+Y)},
\ee
since for all random vectors supported in $M_1+M_2$, the uniform distribution on $M_1+M_2$ maximizes the entropy \cite{CT06:book}.
On the other hand, entropy power will give us the lower bound:
\be\label{eq:localsum2}
e^{\frac{1}{n}h(X+Y)}\geq\sqrt{|M_1|^{\frac{2}{n}}+|M_2|^{\frac{2}{n}}}.
\ee
Combining the inequalities \eqref{eq:localsum1} and \eqref{eq:localsum2} gives an inequality weaker than the
Brunn-Minkowski inequality:
\be\label{eq:localsum3}\begin{split}
|M_1+M_2|^{\frac{1}{n}} &\geq\sqrt{|M_1|^{\frac{2}{n}}+|M_2|^{\frac{2}{n}}} \\
&=|M_2|^{\frac{1}{n}}\sqrt{1+\lambda^2} \\
&=|M_2|^{\frac{1}{n}}(1+\frac{1}{2}\lambda^2+o(\lambda^2)),
\end{split}\ee
where $\lambda=\big(\frac{|M_1|}{|M_2|}\big)^{\frac{1}{n}}$. In contrast, the Brunn-Minkowski inequality gives
\be\label{eq:localsum4}
|M_1+M_2|^{\frac{1}{n}}\geq|M_2|^{\frac{1}{n}}(1+\lambda).
\ee
It is well known and easy to see that \eqref{eq:localsum4} implies the following isoperimetric inequality:
\be\label{eq:isoperi}
\liminf_{\epsilon\downarrow0}\frac{|\epsilon M_1+M_2|-|M_2|}{\epsilon}
\geq n|M_2|\bigg(\frac{|M_1|}{|M_2|}\bigg)^{\frac{1}{n}},
\ee
where the equality holds when $M_1$ and $M_2$ are homothetic convex bodies. On the other hand, \eqref{eq:localsum3} will only give
\ben
\liminf_{\epsilon\downarrow0}\frac{|\epsilon M_1+M_2|-|M_2|}{\epsilon} \geq 0.
\een
Hence, when $\lambda$ is small, \eqref{eq:localsum3} is especially poor compared to \eqref{eq:localsum4}
and is not reflecting the correct behavior as $\lambda$ goes to $0$.
But note that when $n=1$, Corollary~\ref{cor:dim1} implies the following entropic isoperimetric inequality:
\be\label{eq:entisoperi}
\liminf_{\epsilon\downarrow0}\frac{h(Y+\epsilon X)-h(Y)}{\epsilon}
\geq \frac{1}{2}\frac{|M_1|}{|M_2|},
\ee
with equality when $M_1$ and $M_2$ are symmetric intervals.
This implies
\ben
\liminf_{\epsilon\downarrow0}\frac{e^{h(Y+\epsilon X)}-e^{h(Y)}}{\epsilon}\geq\frac{1}{2}|M_2|\frac{|M_1|}{|M_2|}.
\een
Using \eqref{eq:localsum1} again, we get
\ben
\liminf_{\epsilon\downarrow0}\frac{|\epsilon M_1+M_2|-|M_2|}{\epsilon}\geq\frac{1}{2}|M_2|\frac{|M_1|}{|M_2|},
\een
which, although still weaker than \eqref{eq:isoperi}, is reflecting the correct behavior.
\end{rmk}

\begin{rmk}
If $X$ is a uniform distribution on a symmetric interval and if we define
\ben
\widetilde{I}(Y)=\liminf_{\epsilon\downarrow0}\frac{h(Y+\epsilon X)-h(Y)}{\epsilon},
\een
then the inequality \eqref{eq:entisoperi} can be rewritten as
\ben
\widetilde{I}(Y)\geq\widetilde{I}(Y^*).
\een
This is very similar to Corollary~\ref{cor:fishdec}, but with the role of Gaussian replaced by a uniform.
\end{rmk}

\begin{rmk}
In the case where the sets under consideration are convex, much better bounds can be given using
the fact that the convolution of the uniforms yields a $\kappa$-concave measure for $\kappa>0$.
More details, including the definition of $\kappa$-concavity, can be found in \cite{BM12:jfa}.
\end{rmk}

\section{Another application: Entropy of Marginals of L\'{e}vy Processes}
\label{sec:levy}

In this section, we develop some simple applications to L\'{e}vy processes of our results. The main result is Theorem~\ref{thm:levymain}, whose method of proof is very similar to that of \cite[Theorem 1.1]{BM10} and \cite[Lemma 3.2]{DSS13} (the latter deals with the notion of symmetric rearrangements around infinity instead of the usual symmetric rearrangements, though). The key is that L\'{e}vy processes are weak limits of compound Poisson processes. In fact, the inequality \eqref{eq:alreadyknown}, an intermediate step in the proof of our Theorem~\ref{thm:levymain}, is readily implied by \cite[Theorem 1.1]{BM10}, although we give a full proof here for completeness.

\begin{prop}\label{prop:keyofsection}
Suppose the $n$-dimensional process $\{X_t:t\geq 0\}$ can be represented as
\bee
X_t=A^{\frac{1}{2}}\cdot B_t+\sum_{i=1}^{N_t}Y_i,
\ene
where $A$ is a $n$ by $n$ symmetric and strictly positive definite matrix, $B_t$ is the standard $n$ dimensional Brownian motion, $N_t$ is the Poisson process with rate $\lambda\geq0$, independent of the process $B_t$, and $Y_i$s, independent of the processes $B_t$ and $N_t$, are i.i.d. random vectors with density $f$.

We now define the rearranged process $Z_t$ to be:
\bee
Z_t=\text{det}^{\frac{1}{2n}}(A)B_t+\sum_{i=1}^{N_t}Y^*_i,
\ene
where $B_t$ and $N_t$ are as before and $Y^*_i$s, independent of the processes $B_t$ and $N_t$, are i.i.d. random vectors with density $f^*$.

We then have:
\bee
h_p(X_t)\geq h_p(Z_t),
\ene
for $t>0$ and $0< p\leq+\infty$.
\end{prop}
\begin{proof}
Let the density of $A^{\frac{1}{2}}\cdot B_t$ be $g_t$, the density of $X_t$ be $p_t$ and the density of $Z_t$ be $q_t$. Then
\bee
p_t=e^{-\lambda t}\sum_{k=0}^{\infty}\frac{(\lambda t)^k}{k!}g_t\star\underbrace{f\star f\star\cdot\cdot\cdot\star f}_k,
\ene
and
\bee
q_t=e^{-\lambda t}\sum_{k=0}^{\infty}\frac{(\lambda t)^k}{k!}g^*_t\star\underbrace{f^*\star f^*\star\cdot\cdot\cdot\star f^*}_k.
\ene
It suffices to show $p_t\prec q_t$ by Proposition~\ref{prop:gen-id}. We recall the following representation \cite{Bur09:tut} (which has been used several times in Section~\ref{sec:major}):
\begin{equation}\label{eq:rep}
\int_{B(0,r)}f^*(x)dx=\sup_{|C|=|B(0,r)|}\int_C f(x)dx,
\end{equation}
for any density $f$, where $B(0,r)$ is the open ball with radius $r$.  Hence
\bee
\int_{B(0,r)}p^*_t(x)dx=\sup_{|C|=|B(0,r)|}\sum_{k=0}^{\infty}e^{-\lambda t}\frac{(\lambda t)^k}{k!}\int_C g_t\star\underbrace{f\star f\star\cdot\cdot\cdot\star f}_k dx.
\ene
By Theorem~\ref{thm:BLL},
\bee
\int_C g_t\star\underbrace{f\star f\star\cdot\cdot\cdot\star f}_k dx\leq\int_{C^*} g^*_t\star\underbrace{f^*\star f^*\star\cdot\cdot\cdot\star f^*}_k dx
\ene
Since $C^*=B(0,r)$ by definition, we get
\be\label{eq:alreadyknown}
\int_{B(0,r)}p^*_t(x)dx\leq\sum_{k=0}^{\infty}e^{-\lambda t}\frac{(\lambda t)^k}{k!}\int_{B(0,r)} g^*_t\star\underbrace{f^*\star f^*\star\cdot\cdot\cdot\star f^*}_k dx=\int_{B(0,r)}q_t(x)dx.
\ee
By ~\eqref{eq:rep} again,
\bee
\int_{B(0,r)}q_t
(x)dx\leq\int_{B(0,r)}q^*_t(x)dx.
\ene
We finally get:
\bee
\int_{B(0,r)}p^*_t(x)dx\leq\int_{B(0,r)}q^*_t(x)dx,
\ene
which shows $p_t\prec q_t$.
\end{proof}

From now until the end of this section, we assume that the Brownian part of the standard \levy process $\{X_t\}$ 
is non-degenerate and that the \levy measure $\nu$ is locally absolutely continuous. Hence
\bee
\nu(C)=\int_{C}m(x)dx
\ene
for any Borel set $C\subseteq\mathbb{R}^n\setminus\{0\}$. Since $\nu$ is a \levy measure, we have \cite{Pro05:book}:
\bee
\int_{\mathbb{R}^n\setminus\{0\}}\min(1,|x|^2)m(x)dx<+\infty.
\ene
Note that for $t>0$,
\bee
\int_{\{x:m(x)>t\}}dx\leq1+\frac{1}{t}\int_{|x|\geq1}m(x)dx<+\infty.
\ene
Hence $m^*(x)$ is well defined. We now show $\nu^*(dx)=m^*(x)dx$ also defines a \levy measure.
\begin{lem}\label{lem:levymeasure}
\bee
\int_{\mathbb{R}^n\setminus\{0\}}\min(1,|x|^2)m^*(x)dx<+\infty.
\ene
\end{lem}
\begin{proof}
Define $m_n(x)=m(x)\{x>\frac{1}{n}\}$. Then $\int m_n(x)dx<+\infty$. Since rearrangement preserves $L_p$ norm (Lemma \ref{lem:preser-ent}),
\bee
\int m_n(x)dx=\int m_n^*(x)dx.
\ene
By Lemma~\ref{lem:finvar}, we get:
\bee
\int_{\mathbb{R}^n\setminus\{0\}}\min(1,|x|^2)m_n^*(x)dx\leq\int_{\mathbb{R}^n\setminus\{0\}}\min(1,|x|^2)m_n(x)dx\leq\int_{\mathbb{R}^n\setminus\{0\}}\min(1,|x|^2)m(x)dx
\ene
Note that
\bee
m_n^*(x)=\int_0^{\infty}\mathbb{I}_{(S_t\cap\{x:x>\frac{1}{n}\})^*}(x)dt,
\ene
where $S_t=\{x:m(x)>t\}$. By monotone convergence,
\bee
\bigg|\bigg(S_t\cap\bigg\{x:x>\frac{1}{n}\bigg\}\bigg)^*\bigg|
=\bigg|S_t\cap\bigg\{x:x>\frac{1}{n}\bigg\}\bigg| 
\,\,\uparrow \, |S_t|
\,=\, |S^*_t|.
\ene
Since $(S_t\cap\{x:x>\frac{1}{n}\})^*$, $n=1,2,\cdot\cdot\cdot$, are open balls, we must have $\mathbb{I}_{(S_t\cap\{x:x>\frac{1}{n}\})^*}(x)\uparrow\mathbb{I}_{S^*_t}(x)$. By monotone convergence and the definition of rearrangement again, we obtain
\bee
m_n^*(x)\uparrow m^*(x).
\ene
Another application of monotone convergence will give:
\bee
\int_{\mathbb{R}^n\setminus\{0\}}\min(1,|x|^2)m^*(x)dx\leq\int_{\mathbb{R}^n\setminus\{0\}}\min(1,|x|^2)m(x)dx
\ene
\end{proof}
We now recall the \levy-Khinchine formula \cite{Pro05:book}. For any \levy process $\{X_t\}$, we have:
\bee
\mathbb{E}[e^{iu\cdot X_t}]=e^{-t\psi(u)},
\ene
where 
\ben
\psi(u)=\frac{1}{2}(Au,u)-i\gamma\cdot u+\int(1-e^{iu\cdot x}+iu\cdot x\mathbb{I}_\{|x|\leq1\})\nu(dx) .
\een 
We call $(A,\gamma,\nu)$ the \levy triple of $X_t$. We define the rearranged process $Z_t$ to be a \levy process with 
\levy triple $(|A|^{\frac{1}{n}}\mathbb{I}_{n\times n},0,\nu^*)$. We can show the following
\begin{thm}\label{thm:levymain}
Suppose $A$ is non-degenerate, then
\bee
h_p(X_t)\geq h_p(Z_t),
\ene
where $0<p\leq 1$ and $t>0$. Moreover, if $0<t_1<t_2<\cdot\cdot\cdot<t_n$, then
\bee
h(X_{t_1},X_{t_2},\cdot\cdot\cdot,X_{t_n})\geq h(Z_{t_1},Z_{t_2},\cdot\cdot\cdot,Z_{t_n}).
\ene
\end{thm}
To show this, we need two additional lemmas:
\begin{lem}\label{lem:first}
For $p\in(0,1)\cup(1,\infty)$, $h_p(f\star \mu)\geq h_p(f)$, for any two density $f$ and any probability measure $\mu$. If $f$ is bounded, then this is also true for $p=1$.
\end{lem}
\begin{proof}
This fact is well known but let us sketch the simple proof here. We will only prove the case when $p\in(0,1)$. The other case is entirely similar. By Jensen's inequality,
\ben
(f\star\mu)^p(x)\geq\int f^p(x-y)\mu(dy).
\een
Integrating the above with respect to $x$ and rearranging, we get the desired result. Finally, when $f$ is bounded, we may apply Lemma~\ref{lem:key} to take the limit as $p$ goes to $1$.
\end{proof}
\begin{lem}\label{lem:second}
Let $g$ be a non-degenerate Gaussian density. Then for $0<p\leq 1$, $h_p(\mu\star g)$, as a functional (of $\mu$) on the space of all probability measures, is lower semi-continuous with respect to the weak convergence topology.
\end{lem}
\begin{proof}
Let $\mu_n$ be a sequence of probability measures converging weakly to $\mu$ and assume $g\leq C$. Then by definition of weak convergence, for each $x$,
\ben
\mu_n\star g(x)\rightarrow\mu\star g(x).
\een
When $p\neq 1$, by Fatou's lemma,
\ben
\liminf_{n}h_p(\mu_n\star g)\geq h_p(\mu\star g).
\een
When $p=1$, we can apply an argument in \cite{HV05} to conclude. For completeness, let us sketch the argument here. Note that it suffices to show
\ben
\liminf_n-\int\log\bigg(\frac{\mu_n\star g}{C}\bigg)\frac{\mu_n\star g}{C}dx\geq-\int\log\bigg(\frac{\mu\star g}{C}\bigg)\frac{\mu\star g}{C}dx.
\een
But since $\mu_n\star g\leq C$ and $\mu\star g\leq C$, the above also follows from Fatou's lemma.
\end{proof}
Now we finish the proof of Theorem~\ref{thm:levymain}.
\begin{proof}
By the L\'evy-Ito decomposition \cite{Pro05:book}, we can write
\ben
X_t=X^n_t+Y^n_t,
\een
where $X^n$ is a \levy process with \levy triple $(A,\gamma,\nu_n)$, with $\nu_n(dx)=m_n(x)dx=m(x)\{x>\frac{1}{n}\}dx$ and the process $Y^n$ is independent of $X^n$. Clearly, the density of $X^n_t$ is bounded. By Lemma~\ref{lem:first},
\ben
h_p(X_t)\geq h_p(X^n_t).
\een
$X^n_t$ can be written as the sum of a Brownian motion (with constant drift) and an independent compound Poisson process \cite{Pro05:book}. Hence by Proposition~\ref{prop:keyofsection},
\ben
h_p(X^n_t)\geq h_p(Z^n_t),
\een
where $Z^n_t$ is a \levy process with \levy triple $(|A|^{\frac{1}{n}}\mathbb{I}_{n\times n},0,m_n^*(x)dx)$. Again by the \levy-Ito decomposition, we can write
\ben\begin{split}
&Z^n_t=|A|^{\frac{1}{2n}}B_t+U^n_t,\\
&Z_t=|A|^{\frac{1}{2n}}B_t+U_t,
\end{split}
\een
where $B_\cdot$ is a standard Brownian motion independent of $U^n$ and $U$, $U^n$ is a \levy process with \levy triple $(0,0,m_n^*(x)dx)$ and $U$ is a \levy process with \levy triple $(0,0,m^*(x)dx)$.  Finally, since $m_n^*(x)\uparrow m^*(x)$, $U_n(t)$ converges weakly to $U_t$ for each $t$. Hence by Lemma~\ref{lem:second},
\ben
\liminf_nh_p(Z^n_t)\geq h_p(Z_t).
\een
Combining, we get:
\ben
h_p(X_t)\geq h_p(Z_t).
\een
Finally, by the chain rule and the markov property of \levy process, we have for any \levy process $X$,
\ben\begin{split}
h(X_{t_1},X_{t_2},\cdot\cdot\cdot,X_{t_n})=\sum_{i=1}^nh(X_{t_i}|X_{t_{i-1}}).
\end{split}
\een
By the independent and stationary increment property of \levy process and translation invariance of entropy,
\ben
h(X_{t_1},X_{t_2},\cdot\cdot\cdot,X_{t_n})=\sum_{i=1}^nh(X_{t_i-t_{i-1}}).
\een
A similar expression holds for $Z$. Hence we can conclude.
\end{proof}

\section{Yet another proof of the classical Entropy Power Inequality}
\label{sec:pf-epi}

The goal of this section is to give a new proof the entropy power inequality (Theorem~\ref{thm:epi}) starting from the Main Theorem. We comment here that in this section we
actually only need the Main Theorem for $p=1$ and $k=2$. By Remark~\ref{rmk:entconti}, to prove Theorem~\ref{thm:epi}, we can consider the case
when the two densities are bounded, strictly positive and have finite covariance matrices. These will be assumed throughout this section. For convenience, we will use the following well-known equivalent formulation of the entropy power inequality \cite{DCT91}:
\ben
h(\sqrt{\lambda}X+\sqrt{1-\lambda}Y)\geq\lambda h(X)+(1-\lambda)h(Y),
\een
for all $0<\lambda<1$, where $X$ has density $f_1$ and $Y$ has density $f_2$.
By the Main Theorem and Lemma~\ref{lem:preser-ent}, we can assume that $f_1$ and $f_2$
are spherically symmetric decreasing. Note that by Lemma~\ref{lem:preser-ent},
Lemma~\ref{lem:key} and Lemma~\ref{lem:finvar}, if a density $f$ is bounded,
strictly positive and has finite covariance matrix, then so is $f^*$. Hence we will assume
from now that the two densities are spherically symmetric decreasing, bounded, strictly
positive and have finite covariance matrices.

We first show that an EPI comes almost for free if we assume identical distribution.
The case when $\lambda=\frac{1}{2}$ seems to be folklore; we learned it from Andrew Barron
several years ago. For completeness, we sketch the easy proof for all $\lambda$.

\begin{prop}\label{prop:ident}
Fix any $0<\lambda<1$. Suppose that $X$ and $Y$ have the same distribution. Then:
\begin{equation*}
h(\sqrt{\lambda}X+\sqrt{1-\lambda}Y)\geq h(X),
\end{equation*}
\end{prop}

\begin{proof}
By independence, we have:
\begin{equation*}
h(X,Y)=h(X)+h(Y).
\end{equation*}
By spherical symmetry (in fact, we only need central symmetry) and i.i.d. assumption, we have
\begin{equation*}
\sqrt{\lambda}X+\sqrt{1-\lambda}Y=^d\sqrt{1-\lambda}X-\sqrt{\lambda}Y.
\end{equation*}
By the scaling property for entropy,
\begin{equation*}
h(X,Y)=h(\sqrt{\lambda}X+\sqrt{1-\lambda}Y,\sqrt{1-\lambda}X-\sqrt{\lambda}Y).
\end{equation*}
Now we can use subadditivity of entropy to conclude.
\end{proof}
\vspace{.1in}

We now give a slightly involved proof of the full entropy power inequality starting from the Main Theorem.
For notational simplicity, \textbf{we assume $\mathbf{n=1}$ until the end of this section}.
Our proof is inspired by and may be considered as an adaptation of Brascamp and Lieb's proof of
Young's inequality with sharp constant \cite{BL76b}.

First, we do a simple reduction. Since we assumed $f$ and $g$ are bounded symmetric decreasing, we can approximate these densities pointwise
and monotonically from below by symmetric decreasing simple functions of the form $f_n$ and $g_n$:
\ben
f_n=\sum_{i=1}^{k_n}c^n_i \mathbb{I}^n_i,
\een
where $\mathbb{I}^n_i$ are indicators of symmetric finite intervals with $\mathbb{I}^n_i\leq \mathbb{I}^n_{i+1}$ and $c^n_i>0$
(note that $c^n_i>0$ since $f_n$ is decreasing) and a similar expression for $g_n$.
By our assumption, we can show that for fixed $0<\lambda<1$,
\ben\begin{split}
h(\widetilde{f}_n) &\rightarrow h(f), \\
h(\widetilde{g}_n) &\rightarrow h(g), \\
h\bigg(\frac{1}{\sqrt{\lambda}}\widetilde{f}_n\bigg(\frac{\cdot}{\sqrt{\lambda}}\bigg)
\star\frac{1}{\sqrt{1-\lambda}}\widetilde{g}_n\bigg(\frac{\cdot}{\sqrt{1-\lambda}}\bigg)\bigg)
&\rightarrow
h\bigg(\frac{1}{\sqrt{\lambda}}f\bigg(\frac{\cdot}{\sqrt{\lambda}}\bigg)\star\frac{1}{\sqrt{1-\lambda}}g\bigg(\frac{\cdot}{\sqrt{1-\lambda}}\bigg)\bigg),
\end{split}\een
where $\widetilde{f}_n$ and $\widetilde{g}_n$ are normalized versions of $f_n$ and $g_n$.
This is because, as shown by Harrem\"oes and Vignat \cite{HV05},
if a sequence of uniformly bounded densities converges pointwise to
a density, and the first two moments also converge, then one has convergence of entropies.
Hence, without loss of generality, we can assume that $f$ and $g$ are of the following form:
\ben
f=\sum_{i=1}^{k_1}c^1_i \mathbb{I}^1_i,\quad g=\sum_{i=1}^{k_2}c^2_i \mathbb{I}^2_i.
\een
where $c^1_i>0, c^2_i>0$.

The main trick is to use tensorization, or what physicists call the replica method.
Consider $X_1,X_2,\cdot\cdot\cdot,X_M$, which are $M$ independent copies of $X$,
and independent of these, $Y_1,Y_2,\cdot\cdot\cdot,Y_M$, which are $M$ independent
copies of $Y$. Let $\mathbf{X}=(X_1,X_2,\cdot\cdot\cdot,X_M)$, $\mathbf{Y}=(Y_1,Y_2,\cdot\cdot\cdot,Y_M)$. The densities of $\mathbf{X}$ and $\mathbf{Y}$ are
\ben\begin{split}
F(x_1,x_2,\cdot\cdot\cdot,x_M)&=\prod_{i=1}^Mf(x_i), \\
G(y_1,y_2,\cdot\cdot\cdot,y_M)&=\prod_{i=1}^Mg(y_i).
\end{split}\een

Next, we show that $F^*$ and $G^*$ are both finite mixtures and the number of densities in the mixture grows
at most polynomially in $M$. It is easy to see $F$ takes at most $(M+1)^{k_1}$ values and $G$ takes at most
$(M+1)^{k_2}$ values \cite{BL76b}. Hence just by looking at the definitions of rearrangements, one sees that
$F^*$ takes at most $(M+1)^{k_1}$ values and $G^*$ takes at most $(M+1)^{k_2}$ values. This allows us to
express $F^*$ and $G^*$ as, using the spherically symmetric decreasing property,
\ben\begin{split}
F^*&=\sum_{i=1}^{(M+1)^{k_1}}b^1_i \mathbb{I}_{\eta^1_i}, \\
G^*&=\sum_{j=1}^{(M+1)^{k_2}}b^2_j \mathbb{I}_{\eta^2_j},
\end{split}\een
where $b^1_i>0,b^2_j>0$ and $\mathbb{I}_{\eta^1_i}, \mathbb{I}_{\eta^2_j}$ are indicators of $M$ dimensional balls $\eta^1_i$ and $\eta^2_j$,
centered at the origin and $|\eta^1_i|\leq |\eta^1_{i+1}|, \,|\eta^2_j|\leq|\eta^2_{j+1}|$. Since both are probability densities, we have
\ben\begin{split}
\sum_{i=1}^{(M+1)^{k_1}}b^1_i|\eta^1_i| &=1, \\
\sum_{j=1}^{(M+1)^{k_2}}b^2_j|\eta^2_j| &=1, \\
\sum_{i,j}b^1_ib^2_j|\eta^1_i||\eta^2_j| &=1.
\end{split}\een
We rewrite $F^*$ and $G^*$ as
\ben\begin{split}
F^*&=\sum_{i=1}^{(M+1)^{k_1}} b^1_i |\eta^1_i| \frac{\mathbb{I}_{\eta^1_i}}{|\eta^1_i|}, \\
G^*&=\sum_{j=1}^{(M+1)^{k_2}} b^2_j |\eta^2_j| \frac{\mathbb{I}_{\eta^2_j}}{|\eta^2_j|}.
\end{split}\een
But $\frac{\mathbb{I}_{\eta^1_i}}{|\eta^1_i|}$ and $\frac{\mathbb{I}_{\eta^2_j}}{|\eta^2_j|}$ are exactly the uniform distributions on the
balls $\eta^1_i$ and $\eta^2_j$. Hence both $F^*$ and $G^*$ are mixtures of uniform distribution on balls.

Two more ingredients, of independent interest, are needed for the proof. The first is the concavity of entropy and the
following simple lemma, which may be thought of as a ``reverse concavity'' property of entropy when taking finite mixtures.

\begin{lem}\label{lem:mix}
Let $f$ be a finite mixture of densities, i.e.,
\ben
f=\sum_{i=1}^n c_i f_i ,
\een
where $c_i$ are nonnegative constants summing to $1$, and $f_i$ are densities on $\RL^n$.
Then
\ben
h(f)\leq\sum_i c_i h(f_i)-\sum_i c_i\log c_i .
\een
In particular, $h(f)\leq\sum_i c_i h(f_i)+\log n$.
\end{lem}

\begin{proof}
By definition,
\ben
h(f)=-\sum_ic_i\int f_i\log(f)dx.
\een
Since $\log$ is increasing,
\ben
-\log(f)\leq-\log(c_if_i)=-\log(c_i)-\log(f_i).
\een
Hence
\ben
h(f)\leq-\sum_i\log(c_i)c_i+\sum_ic_ih(f_i).
\een
The term $-\sum_i\log(c_i)c_i$ is exactly the discrete entropy $H(c)$ of $c=(c_1,c_2,\cdot\cdot,c_n)$. Hence
\ben
h(f)\leq H(c)+\sum_ic_ih(f_i)\leq\log(n)+\sum_ic_ih(f_i)
\een
\end{proof}
\vspace{.1in}

\begin{rmk}
For discrete entropy, this lemma is well known and can be found as an exercise in
\cite{CT06:book}; the standard way to prove it is using the data processing inequality
for discrete entropy. For differential entropy (which is our focus), there is no
data processing inequality, and we could not find Lemma~\ref{lem:mix} in
the literature even though it has an extremely simple alternate proof.
\end{rmk}

The last ingredient is the following simple lemma, which is the EPI for uniform distributions on balls,
but with an error term.

\begin{lem}\label{lem:epu}
Let $Z_1$ and $Z_2$ be two independent uniforms on $M$-dimensional balls centered at the origin, then
\ben
h(\sqrt{\lambda}Z_1+\sqrt{1-\lambda}Z_2)\geq\lambda h(Z_1)+(1-\lambda)h(Z_2)+o(M),
\een
where the $o$ symbol is uniform with respect to all pairs of balls centered at the origin.
\end{lem}

\begin{proof}
We will only sketch the argument here. Let
\ben
Z=\sqrt{\lambda}Z_1+\sqrt{1-\lambda}Z_2
\een
and the radii of the balls corresponding to $Z_1$ and $Z_2$ be $b_1$ and $b_2$ respectively.
Suppose the densities of $Z,Z_i$ are $f,f_i,i=1,2$. We now define two $M$-dimensional Gaussian densities:
\ben
g_i(x)={\bigg(\frac{M}{2\pi b^2_i}\bigg)}^{\frac{M}{2}}e^{-\frac{M|x|^2}{2b^2_i}},  i=1,2,
\een
and let $G=\sqrt{\lambda}G_1+\sqrt{1-\lambda}G_2$, with density $g$,
where $G_1$ and $G_2$ are independent random vectors with densities $g_1$ and $g_2$.
We indicate that, using Stirling's approximation, one can show (assuming $M$ even, without loss of generality):
\ben\begin{split}
f_i &\leq\sqrt{\pi M}e^{\frac{1}{12M}}g_i \,, \\
f &\leq\pi Me^{\frac{1}{6M}}g.
\end{split}\een
Hence it is easily seen that
\ben
D(Z\|G)=\lambda D(Z_1\|G_1)+(1-\lambda)D(Z_2\|G_2)+O(\log(M)).
\een
where $D(\cdot\|\cdot)$ is the relative entropy and the $O(\log(M))$ means it is bounded by $\log(M)$ times a universal constant.
Now some easy calculations show that
\ben\begin{split}
D(Z\|G) &=-h(Z)+\frac{M^2}{M+2}-\frac{M}{2}\log\bigg(\frac{M}{2\pi(\lambda b^2_1+(1-\lambda)b^2_2)}\bigg), \\
D(Z_i\|G_i) &=-h(Z_i)+\frac{M^2}{M+2}-\frac{M}{2}\log\bigg(\frac{M}{2\pi b^2_i}\bigg).
\end{split}\een
Hence we get
\be\begin{split}\label{eq:epiasy}
h(Z)-\lambda h(Z_1)-(1-\lambda)h(Z_2) 
&=\frac{M}{2}\bigg(\log(\lambda b^2_1+(1-\lambda)b^2_2)\\
&-\lambda\log(b^2_1)-(1-\lambda)\log(b^2_2)\bigg)+O(\log(M)) \\
&\geq O(\log(M)),
\end{split}\ee
where the last step follows from concavity of $\log$ and the meaning of the $O$ symbol is as before.
\end{proof}
\vspace{.1in}

\begin{rmk}
The above proof gives very strong information about the entropy of the sum of two independent uniform on balls, in high dimensions.
The last equality is in fact equivalent to
\be\label{eq:err}
\frac{2h(Z_1+Z_2)}{M}=\log\big(e^{\frac{2h(Z_1)}{M}}+e^{\frac{2h(Z_2)}{M}}\big)+\frac{O(\log(M))}{M}.
\ee
Hence we obtain the conclusion that asymptotically, EPI becomes an equality for
two independent uniform distribution on balls in high dimensions.
The expression \eqref{eq:err} implies:
\be\label{eq:err2}
e^{\frac{2h(Z_1+Z_2)}{M}}\leq c_M\big(e^{\frac{2h(Z_1)}{M}}+e^{\frac{2h(Z_2)}{M}}\big),
\ee
where $c_M$, depending only on $M$, goes to $1$ as $M$ goes to infinity.
This is not surprising because uniform distributions on high-dimensional balls are close to Gaussians.
We may also note in passing that it was recently shown in \cite{MK14} (and also in \cite{BM13:goetze} with a slightly worse constant)
that for IID, log-concave random vectors $U_1$ and $U_2$ taking values in $\RL^M$,
\ben
e^{\frac{2h(U_1+U_2)}{M}}\leq 2\big(e^{\frac{2h(U_1)}{M}}+e^{\frac{2h(U_2)}{M}}\big) ,
\een
while the more general reverse entropy power inequality of \cite{BM11:cras, BM12:jfa} gives such an inequality
for arbitrary independent  log-concave random vectors (thus covering two balls of different radii), but with a non-explicit constant.
\end{rmk}
\vspace{.1in}

We now complete the proof of the original EPI. Let $\mathbf{Z^1_i}$ and $\mathbf{Z^2_j}$ be
independent random vectors uniformly distributed on the balls $\eta^1_i$ and $\eta^2_j$.
\ben\begin{split}
Mh(\sqrt{\lambda}X_1+\sqrt{1-\lambda}Y_1)  &\stackrel{(a)}{=}h(\sqrt{\lambda}\mathbf{X}+\sqrt{1-\lambda}\mathbf{Y}) \\
&\stackrel{(1)}{\geq} h(\sqrt{\lambda}\mathbf{X}^*+\sqrt{1-\lambda}\mathbf{Y}^*) \\
&\stackrel{(2)}{\geq}\sum_{i=1}^{(M+1)^{k_1}}\sum_{j=1}^{(M+1)^{k_2}}
b^1_ib^2_j|\eta^1_i||\eta^2_j|h(\sqrt{\lambda}\mathbf{Z^1_i}+\sqrt{1-\lambda}\mathbf{Z^2_j}) \\
&\stackrel{(3)}{\geq}\lambda\sum_{i=1}^{(M+1)^{k_1}}b^1_i|\eta^1_i|h(\mathbf{Z^1_i})
+ (1-\lambda)\sum_{j=1}^{(M+1)^{k_2}}b^2_j|\eta^2_j|h(\mathbf{Z^2_j})   -  C\log(M) \\
&\stackrel{(4)}{\geq}\lambda h(\mathbf{X}^*)+(1-\lambda)h(\mathbf{Y}^*)
-C\log(M)-\lambda k_1\log(M+1)-(1-\lambda)k_2\log(M+1) \\
&\stackrel{(b)}{=}\lambda h(\mathbf{X})+(1-\lambda)h(\mathbf{Y})+O(\log(M)) \\
&\stackrel{(c)}{=}M\lambda h(X_1)+M(1-\lambda)h(Y_1)+O(\log(M)) .
\end{split}\een
Here $(a)$ follows from independence; $(1)$ follows from $M$ dimensional version of the Main Theorem; $(2)$ follows from concavity of entropy and the simple fact that convolution of mixtures is a mixture of convolutions; $(3)$ follows from Lemma~\ref{lem:epu}; $(4)$ follows from Lemma~\ref{lem:mix}; $(b)$ follows from Lemma~\ref{lem:ent-pres}; $(c)$ follows again from independence. We finally get:
\ben
Mh(\sqrt{\lambda}X_1+\sqrt{1-\lambda}Y_1)\geq M\lambda h(X_1)+M(1-\lambda)h(Y_1)+O(\log(M)).
\een
Dividing both sides by $M$ and taking the limit as $M$ goes to infinity, we recover the full EPI.

%
\section*{Acknowledgment}
We would like to thank the Department of Statistics at Yale University,
where this research was begun, and the Department of Mathematics at the
Indian Institute of Science, Bangalore, where this research was completed.
We are also indebted to Bruce Hajek. Professor Hajek rediscovered the Riesz rearrangement inequality
while solving an optimization problem for paging and registration in cellular networks with
K.~Mitzel and S.~Yang in \cite{HMY08},
and suggested in a talk in 2008 (that the authors were not fortunate enough to attend)
that it should be interesting to information theorists. We discovered the slides of his talk
after a version of this paper had been written with just the first proof of the Main Theorem,
and they led us to the second, simpler and more general, proof in this paper. We are grateful to
Sergey Bobkov for sharing two preprints of his recent papers \cite{BC13:1, BC13:2}
with G. P. Chistyakov; although we received these after the first version of this paper was written,
we added the discussion in Section~\ref{sec:refine} in response to their results. Finally,
we thank Almut Burchard for valuable feedback on an earlier version of this paper,
as well as for supplying several relevant references.




\begin{thebibliography}{51}
\providecommand{\url}[1]{#1}
\csname url@samestyle\endcsname
\providecommand{\newblock}{\relax}
\providecommand{\bibinfo}[2]{#2}
\providecommand{\BIBentrySTDinterwordspacing}{\spaceskip=0pt\relax}
\providecommand{\BIBentryALTinterwordstretchfactor}{4}
\providecommand{\BIBentryALTinterwordspacing}{\spaceskip=\fontdimen2\font plus
\BIBentryALTinterwordstretchfactor\fontdimen3\font minus
  \fontdimen4\font\relax}
\providecommand{\BIBforeignlanguage}[2]{{%
\expandafter\ifx\csname l@#1\endcsname\relax
\typeout{** WARNING: IEEEtranS.bst: No hyphenation pattern has been}%
\typeout{** loaded for the language `#1'. Using the pattern for}%
\typeout{** the default language instead.}%
\else
\language=\csname l@#1\endcsname
\fi
#2}}
\providecommand{\BIBdecl}{\relax}
\BIBdecl

\bibitem{ATL89}
\BIBentryALTinterwordspacing
A.~Alvino, G.~Trombetti, and P.-L. Lions, ``On optimization problems with
  prescribed rearrangements,'' \emph{Nonlinear Anal.}, vol.~13, no.~2, pp.
  185--220, 1989. [Online]. Available:
  \url{http://dx.doi.org/10.1016/0362-546X(89)90043-6}
\BIBentrySTDinterwordspacing

\bibitem{BM10}
\BIBentryALTinterwordspacing
R.~Ba{\~n}uelos and P.~J. M{\'e}ndez-Hern{\'a}ndez, ``Symmetrization of
  {L}\'evy processes and applications,'' \emph{J. Funct. Anal.}, vol. 258,
  no.~12, pp. 4026--4051, 2010. [Online]. Available:
  \url{http://dx.doi.org/10.1016/j.jfa.2010.02.013}
\BIBentrySTDinterwordspacing

\bibitem{Bar84:tr}
\BIBentryALTinterwordspacing
A.~R. Barron, ``Monotonic central limit theorem for densities,'' Department of
  Statistics, Stanford University, California, Tech. Rep. \#50, 1984. [Online].
  Available:
  \url{http://www.stat.yale.edu/~arb4/publications_files/monotoning%20central%20limit.pdf}
\BIBentrySTDinterwordspacing

\bibitem{Bec75}
W.~Beckner, ``Inequalities in {F}ourier analysis,'' \emph{Ann. of Math. (2)},
  vol. 102, no.~1, pp. 159--182, 1975.

\bibitem{BM11:cras}
S.~Bobkov and M.~Madiman, ``Dimensional behaviour of entropy and information,''
  \emph{C. R. Acad. Sci. Paris S\'er. I Math.}, vol. 349, pp. 201--204,
  F\'evrier 2011.

\bibitem{BM11:it}
------, ``The entropy per coordinate of a random vector is highly constrained
  under convexity conditions,'' \emph{IEEE Trans. Inform. Theory}, vol.~57,
  no.~8, pp. 4940--4954, August 2011.

\bibitem{BM12:jfa}
\BIBentryALTinterwordspacing
------, ``Reverse {B}runn-{M}inkowski and reverse entropy power inequalities
  for convex measures,'' \emph{J. Funct. Anal.}, vol. 262, pp. 3309--3339,
  2012. [Online]. Available: \url{http://arxiv.org/abs/1109.5287}
\BIBentrySTDinterwordspacing

\bibitem{BC12}
S.~G. Bobkov and G.~P. Chistyakov, ``Bounds on the maximum of the density for
  sums of independent random variables (russian),'' \emph{Zapiski Nauchn.
  Semin. POMI}, 2012.

\bibitem{BC13:1}
------, ``Entropy power inequality for the {R\'enyi} entropy,''
  \emph{Preprint}, 2013.

\bibitem{BC13:2}
------, ``On concentration functions of random variables,'' \emph{To appeear in
  JOTP}, 2013.

\bibitem{BCG14:fisher}
S.~G. Bobkov, G.~P. Chistyakov, and F.~G\"otze, ``Fisher information and the
  central limit theorem,'' \emph{Probab. Theory Relat. Fields}, vol. 159, pp.
  1--59, June 2014.

\bibitem{BMW11}
S.~G. Bobkov, M.~Madiman, and L.~Wang, ``Fractional generalizations of {Y}oung
  and {B}runn-{M}inkowski inequalities,'' in \emph{Concentration, Functional
  Inequalities and Isoperimetry}, ser. Contemp. Math., C.~Houdr{\'e},
  M.~Ledoux, E.~Milman, and M.~Milman, Eds., vol. 545.\hskip 1em plus 0.5em
  minus 0.4em\relax Amer. Math. Soc., 2011, pp. 35--53.

\bibitem{BM13:goetze}
S.~G. Bobkov and M.~M. Madiman, ``On the problem of reversibility of the
  entropy power inequality,'' in \emph{Limit Theorems in Probability,
  Statistics, and Number Theory (in honor of Friedrich G\"otze)}, ser. Springer
  Proceedings in Mathematics and Statistics, P.~E. et~al., Ed.\hskip 1em plus
  0.5em minus 0.4em\relax Springer-Verlag, 2013, vol.~42, available online at
  {\tt http://arxiv.org/abs/1111.6807}.

\bibitem{BL76b}
H.~J. Brascamp and E.~H. Lieb, ``Best constants in {Y}oung's inequality, its
  converse, and its generalization to more than three functions,''
  \emph{Advances in Math.}, vol.~20, no.~2, pp. 151--173, 1976.

\bibitem{BLL74}
H.~J. Brascamp, E.~H. Lieb, and J.~M. Luttinger, ``A general rearrangement
  inequality for multiple integrals,'' \emph{J. Functional Analysis}, vol.~17,
  pp. 227--237, 1974.

\bibitem{Bur94:phd}
A.~Burchard, ``Cases of equality in the {R}iesz rearrangement inequality,''
  Ph.D. dissertation, Georgia Institute of Technology, Atlanta, USA, 1994.

\bibitem{Bur09:tut}
\BIBentryALTinterwordspacing
------. (2009, June) A short course on rearrangement inequalities. [Online].
  Available: \url{http://www.math.utoronto.ca/almut/rearrange.pdf}
\BIBentrySTDinterwordspacing

\bibitem{BS01}
\BIBentryALTinterwordspacing
A.~Burchard and M.~Schmuckenschl{\"a}ger, ``Comparison theorems for exit
  times,'' \emph{Geom. Funct. Anal.}, vol.~11, no.~4, pp. 651--692, 2001.
  [Online]. Available: \url{http://dx.doi.org/10.1007/PL00001681}
\BIBentrySTDinterwordspacing

\bibitem{CS91}
E.~A. Carlen and A.~Soffer, ``Entropy production by block variable summation
  and central limit theorems.'' \emph{Comm. Math. Phys.}, vol. 140, 1991.

\bibitem{Cho74}
K.~M. Chong, ``Some extensions of a theorem of {H}ardy, {L}ittlewood and
  {P}\'olya and their applications,'' \emph{Canad. J. Math.}, vol.~26, pp.
  1321--1340, 1974.

\bibitem{CHV03}
J.~Costa, A.~Hero, and C.~Vignat, ``On solutions to multivariate maximum
  alpha-entropy problems,'' \emph{Lecture Notes in Computer Science}, vol.
  2683, no. EMMCVPR 2003, Lisbon, 7-9 July 2003, pp. 211--228, 2003.

\bibitem{Cos85b}
M.~Costa, ``A new entropy power inequality,'' \emph{IEEE Trans. Inform.
  Theory}, vol.~31, no.~6, pp. 751--760, 1985.

\bibitem{CT06:book}
T.~M. Cover and J.~A. Thomas, \emph{Elements of information theory},
  2nd~ed.\hskip 1em plus 0.5em minus 0.4em\relax Hoboken, NJ:
  Wiley-Interscience [John Wiley \& Sons], 2006.

\bibitem{DCT91}
A.~Dembo, T.~Cover, and J.~Thomas, ``Information-theoretic inequalities,''
  \emph{IEEE Trans. Inform. Theory}, vol.~37, no.~6, pp. 1501--1518, 1991.

\bibitem{DZ98:book}
A.~Dembo and O.~Zeitouni, \emph{Large Deviations Techniques And Applications},
  2nd~ed.\hskip 1em plus 0.5em minus 0.4em\relax New York: Springer-Verlag,
  1998.

\bibitem{DSS13}
A.~Drewitz, P.~Sousi, and R.~Sun, ``Symmetric rearrangements around infinity
  with applications to {L}{\'e}vy processes,'' \emph{Probab. Theory Relat.
  Fields}, 2013.

\bibitem{Gro75}
L.~Gross, ``Logarithmic {S}obolev inequalities,'' \emph{Amer. J. Math.},
  vol.~97, no.~4, pp. 1061--1083, 1975.

\bibitem{HMY08}
\BIBentryALTinterwordspacing
B.~Hajek, K.~Mitzel, and S.~Yang, ``Paging and registration in cellular
  networks: jointly optimal policies and an iterative algorithm,'' \emph{IEEE
  Trans. Inform. Theory}, vol.~54, no.~2, pp. 608--622, 2008. [Online].
  Available: \url{http://dx.doi.org/10.1109/TIT.2007.913566}
\BIBentrySTDinterwordspacing

\bibitem{HLP29}
G.~H. Hardy, J.~E. Littlewood, and G.~P\'olya, ``{Some simple inequalities
  satisfied by convex functions},'' \emph{Messenger of Mathematics}, vol.~58,
  pp. 145--152, 1929.

\bibitem{HLP88:book}
G.~H. Hardy, J.~E. Littlewood, and G.~P{\'o}lya, \emph{Inequalities}, ser.
  Cambridge Mathematical Library.\hskip 1em plus 0.5em minus 0.4em\relax
  Cambridge: Cambridge University Press, 1988, reprint of the 1952 edition.

\bibitem{HV05}
P.~Harremo{\"e}s and C.~Vignat, ``A short information theoretic proof of
  {CLT},'' \emph{Unpublished}, 2005.

\bibitem{JV07}
O.~Johnson and C.~Vignat, ``Some results concerning maximum {R}\'enyi entropy
  distributions,'' \emph{Ann. Inst. H. Poincar\'e Probab. Statist.}, vol.~43,
  no.~3, pp. 339--351, 2007.

\bibitem{Joh04:book}
O.~Johnson, \emph{Information theory and the central limit theorem}.\hskip 1em
  plus 0.5em minus 0.4em\relax London: Imperial College Press, 2004.

\bibitem{LL01:book}
E.~H. Lieb and M.~Loss, \emph{Analysis}, 2nd~ed., ser. Graduate Studies in
  Mathematics.\hskip 1em plus 0.5em minus 0.4em\relax Providence, RI: American
  Mathematical Society, 2001, vol.~14.

\bibitem{LLYZ13}
E.~Lutwak, S.~Lv, D.~Yang, and G.~Zhang, ``Affine moments of a random vector,''
  \emph{IEEE Trans. Inform. Theory}, vol.~59, no.~9, pp. 5592--5599, September
  2013.

\bibitem{LYZ07a}
E.~Lutwak, D.~Yang, and G.~Zhang, ``Moment-entropy inequalities for a random
  vector,'' \emph{IEEE Trans. Inform. Theory}, vol.~53, no.~4, pp. 1603--1607,
  2007.

\bibitem{MK14}
M.~Madiman and I.~Kontoyiannis, ``The {R}uzsa divergence for random elements in
  locally compact abelian groups,'' \emph{Preprint}, 2014.

\bibitem{Pro05:book}
P.~E. Protter, \emph{Stochastic integration and differential equations}, ser.
  Stochastic Modelling and Applied Probability.\hskip 1em plus 0.5em minus
  0.4em\relax Berlin: Springer-Verlag, 2005, vol.~21, second edition. Version
  2.1, Corrected third printing.

\bibitem{Rie30}
F.~Riesz, ``Sur une in{\'e}galit{\'e} int{\'e}grale,'' \emph{J. London Math.
  Soc.}, vol.~5, pp. 162--168, 1930.

\bibitem{Rog57}
C.~A. Rogers, ``A single integral inequality,'' \emph{J. London Math. Soc.},
  vol.~32, pp. 102--108, 1957.

\bibitem{Rog87:1}
B.~A. Rogozin, ``An estimate for the maximum of the convolution of bounded
  densities,'' \emph{Teor. Veroyatnost. i Primenen.}, vol.~32, no.~1, pp.
  53--61, 1987.

\bibitem{Rud87:book}
W.~Rudin, \emph{Real and Complex Analysis}.\hskip 1em plus 0.5em minus
  0.4em\relax New York: McGraw-Hill, 1987.

\bibitem{ST14}
G.~Savar{\'e} and G.~Toscani, ``The concavity of {R}{\`e}nyi entropy power,''
  \emph{IEEE Trans. Inform. Theory}, vol.~60, no.~5, pp. 2687--2693, May 2014.

\bibitem{Sha48}
C.~Shannon, ``A mathematical theory of communication,'' \emph{Bell System Tech.
  J.}, vol.~27, pp. 379--423, 623--656, 1948.

\bibitem{Sob38}
S.~L. Sobolev, ``On a theorem of functional analysis,'' \emph{Mat. Sb. (N.S.)},
  vol.~4, pp. 471--497, 1938, amer. Math. Soc. (Transl.) (2) 34 (1963) 39--68.

\bibitem{Sta59}
A.~Stam, ``Some inequalities satisfied by the quantities of information of
  {F}isher and {S}hannon,'' \emph{Information and Control}, vol.~2, pp.
  101--112, 1959.

\bibitem{SV00}
\BIBentryALTinterwordspacing
S.~J. Szarek and D.~Voiculescu, ``Shannon's entropy power inequality via
  restricted {M}inkowski sums,'' in \emph{Geometric aspects of functional
  analysis}, ser. Lecture Notes in Math.\hskip 1em plus 0.5em minus 0.4em\relax
  Berlin: Springer, 2000, vol. 1745, pp. 257--262. [Online]. Available:
  \url{http://dx.doi.org/10.1007/BFb0107219}
\BIBentrySTDinterwordspacing

\bibitem{Tos13}
\BIBentryALTinterwordspacing
G.~Toscani, ``An information-theoretic proof of {N}ash's inequality,''
  \emph{Atti Accad. Naz. Lincei Cl. Sci. Fis. Mat. Natur. Rend. Lincei (9) Mat.
  Appl.}, vol.~24, no.~1, pp. 83--93, 2013. [Online]. Available:
  \url{http://dx.doi.org/10.4171/RLM/645}
\BIBentrySTDinterwordspacing

\bibitem{EH12}
T.~van Erven and P.~Harremo{\"e}s, ``R{\'e}nyi divergence and
  {K}ullback-{L}eibler divergence,'' \emph{Preprint}, 2012, available online at
  {\tt http://arxiv.org/abs/1206.2459}.

\bibitem{WWM14:isit}
L.~Wang, J.~O. Woo, and M.~Madiman, ``A lower bound on the {R{\'e}nyi} entropy
  of convolutions in the integers,'' in \emph{Proc. IEEE Intl. Symp. Inform.
  Theory, to appear}, Honolulu, Hawaii, July 2014.

\bibitem{Wat83}
\BIBentryALTinterwordspacing
T.~Watanabe, ``The isoperimetric inequality for isotropic unimodal {L}\'evy
  processes,'' \emph{Z. Wahrsch. Verw. Gebiete}, vol.~63, no.~4, pp. 487--499,
  1983. [Online]. Available: \url{http://dx.doi.org/10.1007/BF00533722}
\BIBentrySTDinterwordspacing

\end{thebibliography}





\end{document}